\documentclass[12pt, draft]{article}

\usepackage{latexsym,amsmath,amscd,amssymb,graphics, amsthm}
\usepackage{enumerate,framed,color}
\usepackage{verbatim}
\usepackage{graphicx}
\usepackage[numbers]{natbib}

\usepackage{hyperref}
\usepackage{url}
\usepackage[margin=1in]{geometry}

\usepackage[all]{xy}

\newcommand{\rem}[1]{}

\DeclareMathOperator{\ad}{ad}

\makeatletter

\@addtoreset{figure}{section}
\def\thefigure{\thesection.\@arabic\c@figure}
\def\fps@figure{h, t}
\@addtoreset{table}{bsection}
\def\thetable{\thesection.\@arabic\c@table}
\def\fps@table{h, t}
\@addtoreset{equation}{section}

\makeatother

\rem{
\textwidth 6.8 truein
\oddsidemargin -0.2 truein
\evensidemargin 0  truein
\topmargin -0.5 truein
\textheight 9.5 in
}

\newcommand\possesivecite[1]{\citeauthor{#1}'s \citep{#1}}

\begin{document}

\allowdisplaybreaks

\newtheorem{theorem}{Theorem}[section]
\newtheorem{definition}[theorem]{Definition}
\newtheorem{lemma}[theorem]{Lemma}
\newtheorem{remark}[theorem]{Remark}
\newtheorem{proposition}[theorem]{Proposition}
\newtheorem{corollary}[theorem]{Corollary}
\newtheorem{example}[theorem]{Example}

\newcommand{\pair}[2]{{\big\langle {#1}\, , \, {#2}\big\rangle}}
\newcommand{\Pair}[2]{{\Big\langle {#1}\, , \, {#2}\Big\rangle}}
\newcommand{\ppair}[2]{{\bigg\langle {#1}\, , \, {#2}\bigg\rangle}}
\newcommand{\Lpair}[2]{{\big( {#1}\, \big| \, {#2}\big)}}
\newcommand{\LPair}[2]{{\Big( {#1}\, \Big| \, {#2}\Big)}}
\newcommand{\Lppair}[2]{{\bigg( {#1}\, \bigg| \, {#2}\bigg)}}

\def\below#1#2{\mathrel{\mathop{#1}\limits_{#2}}}

\newcommand{\lform}[2]{{\big( {#1} \big|\, {#2}\big)}}
\newcommand{\Lform}[2]{{\Big( {#1} \Big|\, {#2}\Big)}}
\newcommand{\scp}[2]{{\Big\langle {#1}\, , \, {#2}\Big\rangle}}
\newcommand{\prt}{\partial}
\newcommand{\sig}{\sigma}
\newcommand{\mR}{{\mathbb{R}}}
\newcommand{\id}{{\mathrm{id}}}
\newcommand{\ti}{\times}
\newcommand{\Ad}{\text{Ad}}
\renewcommand{\ad}{\text{ad}}
\newcommand{\de}{\delta}
\newcommand{\om}{\omega}
\newcommand{\al}{\alpha}
\newcommand{\ga}{\gamma}
\providecommand{\th}{}
\renewcommand{\th}{\theta}
\newcommand{\la}{\lambda}
\newcommand{\be}{\beta}
\newcommand{\Om}{\Omega}
\newcommand{\mg}{{\mathfrak g}}
\newcommand{\CX}{{\mathcal X}}
\newcommand{\cst}{\text{cst}}
\newcommand{\ze}{\zeta}
\newcommand{\gd}{\dot{g}}
\newcommand{\deta}{\dot{\eta}}
\newcommand{\dxi}{\dot{\xi}}
\newcommand{\gp}{g^\prime}
\newcommand{\ginv}{{g}^{-1}}
\newcommand{\bfu}{\mathbf{u}}
\newcommand{\bfw}{\mathbf{w}}
\newcommand{\zhat}{\hat{z}}
\newcommand{\beq}{\begin{equation}}
\newcommand{\beqs}{\begin{equation*}}
\newcommand{\beqa}{\begin{eqnarray}}
\newcommand{\beqas}{\begin{eqnarray*}}
\newcommand{\eeq}{\end{equation}}
\newcommand{\eeqs}{\end{equation*}}
\newcommand{\eeqa}{\end{eqnarray}}
\newcommand{\eeqas}{\end{eqnarray*}}

\pagestyle{myheadings}

\markright{Bruveris et al.  \hfill  {The momentum map representation of images} \hfill  \qquad}

\title{The momentum map representation of images}
\author{M. Bruveris$^{1}$, F. Gay-Balmaz$^{2}$, D. D. Holm$^{1}$,  and T. S. Ratiu$^{3}$}
\addtocounter{footnote}{1}
\footnotetext{Department of Mathematics, Imperial College London. London SW7 2AZ, UK. Partially supported by Royal Society of London, Wolfson Award.
\texttt{m.bruveris08@imperial.ac.uk, d.holm@imperial.ac.uk}
\addtocounter{footnote}{1} }
\footnotetext{Control and Dynamical Systems, California Institute of Technology 107-81, Pasadena, CA 91125, USA and Laboratoire de 
M\'et\'eorologie Dynamique, \'Ecole Normale Sup\'erieure/CNRS, Paris, France. Partially supported by a Swiss NSF grant.
\texttt{fgbalmaz@cds.caltech.edu}
\addtocounter{footnote}{1} }
\footnotetext{Section de
Math\'ematiques and Bernoulli Center, \'Ecole Polytechnique F\'ed\'erale de
Lausanne.
CH--1015 Lausanne. Switzerland. Partially supported by a Swiss NSF grant.
\texttt{Tudor.Ratiu@epfl.ch}
\addtocounter{footnote}{1} }

\date{Revised version August 1, 2010}
\maketitle

\makeatother

\maketitle

\begin{abstract}
This paper discusses the mathematical framework for designing methods of large deformation matching (LDM) for image registration in computational anatomy. After reviewing the geometrical framework of LDM image registration methods, a theorem is proved showing that these methods may be designed by using the actions of diffeomorphisms on the image data structure to define their associated momentum representations as (cotangent lift) momentum maps.  To illustrate its use, the momentum map theorem is shown to recover the known algorithms for matching landmarks, scalar images and vector fields. After briefly discussing the use of this approach for Diffusion Tensor (DT) images, we explain how to use momentum maps in the design of registration algorithms for more general data structures. For example, we extend our methods to determine the corresponding momentum map for registration using semidirect product groups, for the purpose of matching images at two different length scales. Finally, we discuss the use of momentum maps in the design of image registration algorithms when the image data is defined on manifolds instead of vector spaces.
\end{abstract}

\tableofcontents

\section{Introduction}
\label{sec-Intro}

Large deformation diffeomorphic matching methods (LDM) for image registration are based on minimizing the sum of a \emph{kinetic energy metric}, plus a \emph{penalty term}. The former ensures that the deformation follows an optimal path, while the latter ensures an acceptable tolerance in image mismatch. The LDM approaches were introduced and systematically developed in \citet{Trouve1995, Trouve1998, DuGrMi1998, JoMi2000, MiYo2001, Beg2003}, and \citet{Begetal2005}.  See \citet{MiTrYo2002} for an extensive review of this development. The LDM approach fits within \possesivecite{Gr1993} \emph{deformable template} paradigm for image registration. Grenander's paradigm, in turn, is a development of the biometric strategy introduced  by \citet{Thompson1917} of comparing a template image $I_0$ to a target image $I_1$ by finding a smooth transformation that maps the template to the target. This transformation is assumed to belong to a Lie group $G$ that acts on the set of images $V$ containing $I_0$ and $I_1$. The effect of the transformation on the data structure is called the \emph{action} $G\times V\to V$ of the Lie group $G$ on the set $V$. For example, the action of $g\in G$ on $I_0\in V$ is denoted as $g I_0\in V$. \medskip

The objective of LDM is not just to determine a deformation $g_1 \in G$ such that the group action $g_1 I_0$ of $g_1\in G$ on the template $I_0\in V$ approximates the target $I_1\in V$ to within a certain tolerance. Rather, the objective of LDM is to find the \emph{optimal} path $g_t\in G$ continuously parametrized by time $t\in\mathbb{R}$ that smoothly deforms $I_0$ through $I_t = g_t I_0$ to $g_1 I_0$. The optimal path $g_t\in G$ is defined as the path that costs the least in time-integrated kinetic energy for a given tolerance. Hence, the deformable template method may be formulated as an optimization problem based on a trade-off between the following two properties: (i) the tolerance for inexact matching between the final deformed template $g_1 I_0$ and the target template $I_1$; and (2) the cost of time-integrated kinetic energy of the rate of deformation along the path $g_t$. The former is defined by assigning a norm $\lVert\,\cdot\,\rVert: V\to\mathbb{R}$ to measure the mismatch $\lVert g_t I_0-I_1\rVert$ between the two images. The latter is obtained by choosing a Riemannian metric $|\,\cdot\,|: TG\to\mathbb{R}$ that defines the kinetic energy on the tangent space $TG$ of the group $G$. In this setting, a notion of distance between two images emerges, that allows one to compare similarity of images in terms of transformations. This is the setting for the development of computational anatomy using the inexact template matching approach for the registration of images. For more details and background about LDM, see \citet{MiYo2001, MiTrYo2002, Beg2003}, and \citet{Begetal2005}.
\medskip

In applications of LDM to the analysis of features in bio-medical images, the optimal path $g_t$ is naturally chosen from among the diffeomorphic transformations $G={\rm Diff}(\Omega)$ of an open, bounded domain $\Omega$. The domain $\Omega$ will be taken to be the ambient space in which the anatomy is located. Recall that a \emph{diffeomorphism} $g\in {\rm Diff}(\Omega)$ is a smooth invertible map (i.e., a invertible function that maps the domain $\Omega$ onto itself) whose inverse is also smooth. The one-to-one property of these transformations ensures that disjoint sets remain disjoint, so that, e.g.,  no fusion of points occurs under LDM. Continuity of the diffeomorphisms ensures that connected sets remain connected. Smoothness of these transformations ensures preservation of the smoothness of boundaries of the anatomical objects in bio-medical images. The invertibility of diffeomorphisms and their stability under composition also allows one to regard $\rm Diff(\Omega)$ formally as a Lie group.\medskip

Different types of bio-medical images contain various types of information that may be represented in a number of geometrically different types of data structures. For example, the data structures for MR images are scalar functions, or densities, while data obtained DT-MRI can be represented as symmetric tensor fields. Naturally, the design of image registration algorithms based on the theory of transformations must take differences in data structure into account. \medskip

Registration of DT-MRI data -- necessary for the quantitative analysis of anatomical features such as tissue geometry and local fiber orientation -- is much more complicated than registration of scalar image data. This complication arises because local fiber orientation changes under a diffeomorphic transformation and this reorientation has to be included properly in the design of LDM matching algorithms for DT-MRI. A further complication arises because it is not entirely understood how macroscopic deformation influences microscopic properties such as fiber-orientation and diffusivity of water. Though significant efforts have been directed at scalar image registration, little work has been done on matching tensor images using LDM. For the pioneering efforts in the use of LDM with DT-MRI see \citet{Alex-etal-1999, Alex-etal-2001, Cao-etal2005, Cao2006}.\medskip

In summary, the LDM approach models computational anatomy as a deformation of an initial template configuration. The images describing the anatomy are defined on an open bounded set $\Omega$ and the path from the template image $I_0$ to the target image $I_1$ is viewed as the continuous deformation $I_t:=g_t I_0$ under the path of diffeomorphic transformations $g_t\in {\rm Diff}(\Omega)$ acting on the initial  template $I_0$. Importantly, the optimal path of diffeomorphic transformations $g_t$ depends on three main factors: namely, how the action $g_t I_0$ is defined, as well as the definitions of the kinetic energy and the tolerance norm. Images representing different types of information may transform differently under $G={\rm Diff}(\Omega)$. Hence, the optimal path $g_t\in {\rm Diff}(M)$ sought in the LDM approach will depend on the geometrical properties of the data structures that represent the information in the various types of images. \medskip

In the geometrical framework for the LDM approach, the optimal transformation path $g_t\in {\rm Diff}(M)$ may be estimated by using the variational optimization method developed in \citet{Beg2003} and 
\citet{Begetal2005}. Namely, the optimal path for the matching diffeomorphism in this problem may be obtained from a gradient-descent algorithm based on the directional derivative of the cost functional. The cost functional must balance the energy of the deformation path versus the tolerance of mismatch, while taking  proper account of the transformation properties of the image data structure. Other promising methods besides LDM exist, such as the metamorphosis approach discussed in \citet{MiYo2001, TrYo2005} and \citet{HoTrYo2009}.  Metamorphosis is a variant of LDM, that allows the evolution $I_t$ of the image template to deviate from pure deformation. It is also a promising method in the LDM family, but its discussion is beyond our present scope. \medskip

Our aim in this paper is to show that a simple and universal property of  transformation theory, called the \emph{momentum map} can be used to identify and derive the LDM algorithm corresponding to any data structure on which diffeomorphisms may act. That is, the momentum map approach enables one to tailor the LDM algorithm to the transformation properties of the data structure of the images to be matched. For basic introductions to the momentum map in geometric mechanics, see \citet{Ho2008} or \citet{MaRa1999}. For more extensive treatments see \citet{AbMa1978, OrRa2004}.
\medskip

Our interests here focus mainly on deriving the momentum maps corresponding to the various types of data structures, rather than developing the matching dynamics that they subsequently produce. In particular, we shall discuss how one uses the momentum map approach to cope with different data structures, such as densities, vector fields or tensor fields, by recognizing their shared properties in a unified geometrical framework. \medskip

{
The discussion in this paper is mostly on the informal level concerning completions of the diffeomorphism group and the well-posedness of the resulting Euler-Poincar\'e equation.
}

\paragraph{Plan of the paper.} In Section \ref{sec-GeomReg} we begin by discussing the geometry underlying the standard algorithm for LDM introduced in \citet{Begetal2005}. With this motivation we then introduce an abstract framework in which to model registration problems. We derive the equivalent of Beg's formula in the abstract framework in Theorem \ref{delta_E_left} and show that it has the structure of a momentum map. The end of the section is devoted to a discussion of the EPDiff equation and the importance of the initial momentum.

After presenting the abstract framework we apply it in Section \ref{sec-GeomRegDiff} to a range of examples commonly encountered in computational anatomy: landmarks, scalar images, vector fields and symmetric tensor fields arising from DT-MRIs. We emphasize the momentum maps in these examples as the main ingredient in our framework and show how to recover results found in the literature.

Section \ref{sec-GeomReg-SDP} is devoted to a generalization of standard LDM in a different direction. This section takes into account the presence of two different length scales in the image and formulates a version of LDM that uses a semidirect product of two diffeomorphism groups --- one for each length scale --- to perform the registration. We show that for images defined by scalar functions this approach yields a momentum map that is very similar to Beg's formula, except that we use the sum of two kernels, instead  of only one kernel.

Besides the formulation of LDM as in \citet{Begetal2005}, other penalty terms have been proposed by \citet{Beg2007}, \citet{Avants2008} and \citet{ZaHaNi2009}. We show in Section \ref{sec-symm} that these other proposed penalty terms result in momentum map structures that are similar to those in the formulation of Theorem \ref{delta_E_left}.

Our approach can also be generalized to include data structures defined on manifolds that do not possess a linear structure. In Section \ref{nonlin_gen} we consider the extensions of the theory required to deal with data structure defined on manifolds and apply these extensions in examples.

\section{Geometry of Registration}
\label{sec-GeomReg}

\subsection{Motivation}
\label{subsec-Motivation}

The optimal solution to a non-rigid template matching problem is defined as the shortest, or least expensive, path of continuous deformations of one geometric object (template) into another one (target). The goal is the find the path of deformations of the template that is shortest, or costs the least, for a given tolerance in matching the target. The approach focuses its attention on the properties of the action of a Lie group $G$ of transformations on the set of deformable templates.  The attribution of a cost to this process is based on metrics defined on the tangent space $TG$ of the group $G$, following \possesivecite{Gr1993} principles.

\subsubsection*{Formulation of LDM}
In the LDM framework this template matching procedure for image registration is formulated as follows. Suppose an image, say a medical image, is acquired using MRI, CT, or some other imaging technique. To begin, consider the case that the information in an image can be represented as a function $I: \Omega \to \mathbb{R}$, where $\Omega \subseteq \mathbb{R}^d$ is the domain of the image. 
We denote the data structure by writing $I\in V=\mathcal{F}(\Omega)$, the space of smooth functions encoding the information in the images. One usually deals with planar ($d=2$) or volumetric ($d=3$) images. Consider the comparison of two images, consisting of a function $I_0$ representing the template image  and  $I_1$ the target image. The goal is to find a transformation $\phi:\Omega \to \Omega$, such that the transformed image $I_0 \circ \phi^{-1}$ matches the target image $I_1$ with minimal error, as measured by, say, the $L^2$ norm of their difference
\[
E_2(I_0, I_1) = \lVert I_0 \circ \phi^{-1} - I_1 \rVert^2_{L^2}
\,.
\]
For this purpose, one introduces a time-indexed deformation process, that starts at time $t = 0$ with the template (denoted $I_0$), and reaches the target $I_1$ at time $t = 1$.  At a given time $t$ during this process, the current object $I_t$ is assumed to be the image of the template, $I_0$, obtained through a sequence of deformations.  

We also want the time-indexed transformation to be \emph{regular}. To ensure its regularity, we require the transformation to be generated as the flow of a smooth time dependent vector field $u : [0,1]\times \Omega \to \Omega$, i.e. $\phi=\phi_1$ with
\begin{equation}
\label{eq-mot-flow}
\partial_t \phi_t = u_t \circ \phi_t, \quad \phi_0(x)=x.
\end{equation}
We measure the regularity of $u_t$ via a kinetic-energy like term
\[
E_1(u_t) = \int_0^1 \lvert u_t \rvert^2_H dt
\]
where $\lvert u_t \rvert_\mathcal{H}$ is a norm on the space of vector fields on $\Omega$ defined in terms of a positive self-adjoint differential operator $L$ by
\begin{equation}
\label{Lop-def}
\lvert u_t \rvert^2_\mathcal{H} = \langle u, Lu \rangle_{L^2}\,.
\end{equation}
The operator $L$ is commonly chosen as $Lu = u - \alpha^2 \Delta u$. We denote by $\mathcal{H}$ this space of vector fields.

Following \citet{Begetal2005} we can cast the problem of registering $I_0$ to $I_1$ as a variational problem. Namely, we seek to minimize the cost
\begin{equation}
\label{eq-mot-cost}
E(u_t) = \int_0^1 \lvert u_t \rvert^2_\mathcal{H} dt + \frac{1}{2\sigma^2} \lVert I_0 \circ \phi^{-1}_1 - I_1 \rVert^2_{L^2}
\end{equation}
over all time-dependent vector fields $u_t$. The transformation $\phi_1$ is related to the vector field $u_t$ via \eqref{eq-mot-flow}. A necessary condition for a vector field $u_t$ to be minimal is that the derivative of the cost functional $E$ vanishes at $u_t$, that is $D E(u_t)=0$. It is shown in \citet[theorem 2.1.]{Begetal2005} and \citet[theorem 4.1]{MiTrYo2002} that $D E(u_t)=0$ is equivalent to
\begin{equation}
\label{eq-mot-grad}
L u_t = \frac{1}{\sigma^2} \lvert \det D\phi_{t,1}^{-1} \rvert (J^0_t - J^1_t) \nabla J^0_t
\,,
\end{equation}
where $\phi_{t,s}=\phi_t \circ \phi_s^{-1}$ and $J^0_t = I_0 \circ \phi^{-1}_{t,0}$, $J^1_t = I_1 \circ \phi^{-1}_{t,1}$. This condition is then used in \citet{Begetal2005} to devise a gradient descent algorithm for numerically computing the optimal transformation $\phi_1$.

\subsubsection*{Geometric reformulation of LDM}
Formula \eqref{eq-mot-grad} can be reformulated equivalently in a way that emphasizes its geometric nature. As we will show in Section \ref{subsec-AbFrm}, formula \eqref{eq-mot-grad} is equivalent to
\begin{equation}
\label{eq-mot-grad-abstr}
Lu_t = -\frac{1}{\sigma^2} \left(\phi_t\cdot I_0\right) \diamond\left( \phi_{t,1}\cdot \left(\phi_1\cdot I_0 - I_1 \right)^\flat\right)
\,.
\end{equation}
This formula can be understood as follows: the first factor $\phi_t\cdot I_0$ is the action of the transformation $\phi_t$ on the image $I_0\in V=\mathcal{F}(\Omega)$. This is defined as the composition of functions, $\phi_t\cdot I_0 = I_0 \circ \phi_t^{-1}$. The  flat-operator $\flat:V\rightarrow V^*$ maps images in $V$ to the objects in $V^*$ dual to scalar functions, using the inner product on $V$. (These dual objects are the scalar densities.) To describe such an operator, one first needs to choose a convenient space $V^*$ in nondegenerate duality with $V$. We choose to identify $V^*$ with functions in $\mathcal{F}(\Omega)$, by using the $L ^2$-pairing
\[
\left\langle f,I\right\rangle:=\int_\Omega f(x) I(x) dx,
\]
where $dx$ is a fixed volume element on $\Omega$. With this choice, the flat operator $(\,\flat\,)$ is simply the identity map on functions. However, it is important that we conceptually distinguish between elements in $V$ and in its dual $V^*$. Indeed, the action of a transformation $\phi$ on an element in $V^*$ is the dual action, and does not coincide with the action on $V$ in general. 

In our example, the action on $f\in V^*$ is
\begin{equation}
\label{eq-mot-action-den}
\phi\cdot f = \lvert \det D\phi^{-1} \rvert (f \circ \phi^{-1}).
\end{equation}
To see how this action arises, we need the abstract definition of a \emph{dual action}, which is
\[
\langle \phi\cdot f, I \rangle = \langle f, \phi^{-1}\cdot I \rangle.
\]
{
\begin{remark}\label{left-act-inv}{\rm 
The inverse in the definition of the dual action is necessary to ensure that we have a left action: 
\[\phi \cdot(\psi\cdot f) = (\phi\circ \psi)\cdot f\,.\] }
\end{remark}
}
Using this definition and the change of variables formula we see that
\begin{align*}
\langle \phi\cdot f, I \rangle &= \langle f, \phi^{-1}\cdot I \rangle = \int_\Omega (I \circ \phi) f dx = \int_\Omega  I (f \circ \phi^{-1}) \lvert \det D\phi^{-1} \rvert dx \\
&= \left\langle \lvert \det D\phi^{-1} \rvert \left(f \circ \phi^{-1}\right), I \right\rangle
\,.
\end{align*}
Therefore, in the second factor  $\phi_{t,1}\cdot \left(\phi_1\cdot I_0 - I_1 \right)^\flat$ of equation \eqref{eq-mot-grad-abstr}, the term $\left(\phi_1\cdot I_0 - I_1 \right)^\flat$ is interpreted as a function in $V^*$. Consequently, the action is the dual action given by
\[
\phi_{t,1}\cdot  \left(\phi_1\cdot I_0 - I_1 \right)^\flat= \lvert \det D\phi_{t,1}^{-1} \rvert (J^0_t - J^1_t)
\,.
\]

It remains to explain the last ingredient; namely, the diamond map in equation (\ref{eq-mot-grad-abstr}),
\begin{equation}
\diamond: V \times V^\ast \to \mathcal{H}^\ast.
\label{cotlift-momap}
\end{equation}
This is the \emph{cotangent-lift momentum map} associated to the given representation of 
{
the Lie group $G$ on the vector space $V$.
}
Such momentum maps are familiar in geometric mechanics; see, e.g., \citet{Ho2008} or \citet{MaRa1999}. The momentum map (\ref{cotlift-momap}) takes elements of $V \times V^\ast$, regarded as the cotangent bundle $T^\ast V$ of the space of images $V$, to objects in $\mathcal{H}^*$, dual to the vector fields in $\mathcal{H}$. The map $\diamond$ depends on the choice of $\mathcal{H}^*$. For example, using the $L^2$-pairing with respect to the fixed volume element $dx$ and relative to the  Euclidean inner product $(\,\boldsymbol{\cdot}\,)$ in $\mathbb{R}^d$, the momentum map (\ref{cotlift-momap}) is defined for images that are 
{
scalar functions $I \in V=\mathcal{F}(\Omega)$ and densities $f\in V^*=\mathcal{F}^*(\Omega)$ by the relation
}
\begin{equation}
\label{eq-mot-diamond}
\langle  I \diamond f\,,\,u\, \rangle  = \int_\Omega -f \nabla I\boldsymbol{\cdot} u \,dx,
\end{equation}
{
so that in this case $I \diamond f=-f \nabla I$ using the $L^2$ pairing.
}
\begin{remark}[Momentum maps]$\,$
\label{rem-MomMap}\normalfont 

\begin{itemize}
\item
In geometric mechanics, momentum maps generalize the notions of linear and angular momenta. For a mechanical system, whose configuration space is a manifold $M$ acted on by a Lie group $G$, the momentum map $\mathbf{J}: T^\ast M \to \mathfrak{g}^\ast$ assigns to each element of the phase space $T^\ast M$ a generalized ``momentum" in the dual $\mathfrak{g}^*$ of the Lie algebra $\mathfrak{g}$ of the Lie group $G$.  For example, the momentum map for spatial translations is the linear momentum and for rotations it is the angular momentum. 

The importance of the momentum map in geometric mechanics is due to Noether's theorem. Noether's theorem states that the generalized momentum $\mathbf{J}$ is a constant of motion for the system under consideration when its Hamiltonian is invariant under the action of $G$ on $T^*M$. This theorem enables one to turn symmetries of the Hamiltonian into conservation laws.

\item
{
[Notation for momentum maps: $\mathbf{J}$ versus $\diamond\,$] For convenience in
referring to earlier work, e.g., \cite{HoMaRa1998,HoTrYo2009}, we distinguish
between the notation $\mathbf{J}$ for general momentum maps $\mathbf{J}: T^\ast M
\to \mathfrak{g}^\ast$ and the notation $\diamond$ for the particular type of
cotangent-lift momentum maps on linear spaces, $\diamond: V \times V^\ast \to
\mathcal{H}^\ast$ that typically appear in applications of Euler-Poincar\'e theory,
as in equation (\ref{cotlift-momap}). 
}
\end{itemize}
\end{remark}

\begin{remark}[Momentum of images]
\label{rem-ImageMomentum}\normalfont 
Momentum maps for images have been discussed previously. In particular, the momentum map for the EPDiff equation of \citet{HoMa2004} produces an isomorphism between landmarks (and outlines) for images and singular soliton solutions of the EPDiff equation. This momentum map was shown in \citet{HoRaTrYo2004} to provide a complete parameterization of the landmarks by their canonical positions and momenta. A related interpretation of momentum for images in computational anatomy was also discussed in \citet{MiTrYo2006}.

\end{remark}

We now explain in which sense expression \eqref{eq-mot-diamond} is a momentum map. Even though the cost functional \eqref{eq-mot-cost} is not invariant under the action of the diffeomorphism group, one may still define the momentum map $\diamond: V \times V^\ast \to \mathcal{H}^\ast$ via
\[
\langle I \diamond f, u \rangle = \langle f, uI \rangle
\,,
\]
as done in geometric mechanics, see \citet{MaRa1999} and \citet{Ho2008}. The action $uI$ is defined as $uI: = \partial_t \vert_{t=0} \phi_t\cdot I$ for a curve $\phi_t$ such that $\phi_0(x)=x$ and $\partial_t \vert_{t=0} \phi_t = u$. This is the infinitesimal action corresponding to the action of $\operatorname{Diff}(\Omega)$ on $V$. Although the $\diamond$-map does not provide a conserved quantity of the dynamics, it nevertheless helps our intuition and gives us a way to structure the formulas.

Let us apply this concept to image registration for $I\in{\cal F}(\Omega)$, the scalar functions on the domain $\Omega$. The infinitesimal action is given by
\[
uI = \partial_t \big\vert_{t=0} (I \circ \phi_t^{-1}) = -\nabla I \boldsymbol{\cdot} u
\]
and thus the momentum map in this case is
\[
\langle I \diamond f, u \rangle_{V^\ast \times V} = \langle f, -\nabla I \boldsymbol{\cdot} u \rangle = \int_\Omega -(\nabla I \boldsymbol{\cdot} u) f dx
= \langle -f \nabla I, u \rangle_{\mathcal{H}^\ast \times \mathcal{H}}
\,,
\]
as stated in formula \eqref{eq-mot-diamond}. The key is to reinterpret the $L^2$-duality between the functions $-\nabla I \boldsymbol{\cdot} u$ and $f$ as the duality between the vector fields $-f \nabla I$ and $u$.

Using formulas \eqref{eq-mot-diamond} and \eqref{eq-mot-action-den} in equation \eqref{eq-mot-grad-abstr}, we regain the stationarity condition \eqref{eq-mot-grad}.

\begin{remark}\normalfont
Writing the gradient of the cost functional \eqref{eq-mot-grad} in the geometric form \eqref{eq-mot-grad-abstr} has several advantages. For example, it allows us to generalize an algorithm that matches images as scalar functions, to cope with different data structures, such as densities, vector fields, tensor fields and others. Making this generalization allows one to see the underlying common geometrical framework in which we may  unify the treatment of these various data structures. We can also keep the data structure fixed and vary the norm $\lVert\,\cdot\,\rVert$, and thereby alter our criteria of how we measure the distance between two objects.

This geometric framework also enables comparison of different formulations of LDM. For example, one may compare the approach from 
\citet{Begetal2005} presented here with the symmetric approach from \citet{Avants2008} and \citet{Beg2007} and the unbiased approach from \citet{ZaHaNi2009}, in terms of their respective momentum maps.

In addition, the geometrical setting introduced here for image analysis allows us not only to vary the data structure, but also to change the group of transformations. We will explore this possibility in Section \ref{sec-GeomReg-SDP}, when we consider image registration using two diffeomorphism groups simultaneously.
\end{remark}

\subsection{Abstract Framework}
\label{subsec-AbFrm}

Diffeomorphic image registration may be formulated abstractly as follows. Consider a vector space $V$ of \textit{deformable objects} on which an inner product $\langle\,\cdot\,,\,\cdot\,\rangle$ is defined, that allows us to measure distances between two such objects. We can think of $V$ as containing brain MRI images, an example frequently encountered in computational anatomy \citep{MiTrYo2002}. The distance between two objects can be defined as $\lVert I - J \rVert^2 = \langle I-J, I-J \rangle$, which in the case of images is the $L^2$-distance
\[
\int_\Omega \lvert I(x)-J(x) \rvert^2 dx.
\]

The second ingredient is a Lie group $G$ of \textit{deformations}, that acts on the space $V$ of deformable objects from the left
\[
(g,I) \in G \times V \mapsto gI \in V.
\]
In computational anatomy $G$ usually is taken to be the group of diffeomorphisms $\operatorname{Diff}(\Omega)$ or variants of it. A diffeomorphism $\phi \in \operatorname{Diff}(\Omega)$ acts on images by \emph{push-forward}; that is, by pull back by the inverse map,
\[
\phi \cdot I := \phi_\ast I= I \circ \phi^{-1} \quad \textrm{or} \quad \phi\cdot I(x) = I(\phi^{-1}(x)).
\]
Roughly speaking, this action corresponds to drawing the image $I$ on a rubber canvas, then deforming the canvas by $\phi$ and watching the image being deformed along with the canvas. It is also the basis for the familiar \emph{Lagrangian representation} of fluid dynamics as described in \citet{HoMaRa1998}. 

Given a curve $t \mapsto g_t$ of transformations, we define the right-invariant velocity vector $u_t \in \mathfrak{g}$ as
\begin{equation}
\label{eq-abfrm-rv}
u_t = (\partial_t g_t) g_t^{-1}.
\end{equation}
We obtain $u_t$ by taking the tangent vector of $g_t$ and right-translating it back to the tangent space at the identity $T_e G=\mathfrak{g}$, which is the Lie algebra of $G$. Rewriting \eqref{eq-abfrm-rv} as
\begin{equation}
\label{eq-abfrm-rec}
\partial_t g_t = u_t g_t
\end{equation}
and specifying initial conditions at some time $t=s$, we obtain an ordinary differential equation (ODE). If we start with velocity vectors $u_t$, we can solve this ODE to reconstruct the curve $g_t$. This corresponds to the construction of diffeomorphisms as flows of vector fields via the equation
\[
\partial_t \phi_t = u_t \circ \phi_t, \quad \phi_0(x)=x.
\]
Let us denote by $g^u_{t,s}$ the solution of the ODE \eqref{eq-abfrm-rec} rewritten as 
\[
\partial g^u_{t,s} = u_t g^u_{t,s}, \quad g^u_{s,s}=e
\]
with the initial condition that $g^u_{t,s}$ is the identity $e$ at time $t=s$. Since the time $t=0$ will play a special role, we denote $g^u_t := g^u_{t,0}$. Standard results for differential equations show the following properties
\[
g_{t,s} g_{s,r} = g_{t,r}\,, \quad g_{t,s}= g_t g_s^{-1}, \quad g^{-1}_{t,s} = g_{s,t}
\]
which we will use in our calculations.

Following the motivation discussed in Section \ref{subsec-Motivation} we define the abstract version of the cost functional \eqref{eq-mot-cost} as
\begin{equation}
\label{eq-abfrm-cost}
E(u_t):=\int_0^1\ell(u_t)dt+\frac{1}{2\sigma^2}\left\|g^u_1 I_0-I_1 \right\|^2_V
\end{equation}
where the function $\ell:\mathfrak{g} \to \mathbb R$ is a Lagrangian measuring the \emph{kinetic energy} contained in $u_t$ and $\lVert\,\cdot\,\rVert$ is the norm on $V$ induced by the inner product $\langle\,\cdot\,,\,\cdot\,\rangle$. Note that formula \eqref{eq-abfrm-cost} defines a matching problem for any data structure living in a vector space $V$ and any group of deformations $G$ acting on $V$. Although it was inspired by the concrete problem of diffeomorphically matching scalar-valued images, the cost function \eqref{eq-abfrm-cost} no longer contains any reference to image matching.

Next, we want to deduce \eqref{eq-mot-grad-abstr} in our abstract framework. In order to compute the derivative $DE(u_t)$ we need to know how $g_1^u$ behaves under variations $\delta u_t$ of $u_t$. This is answered by the following lemma, the proof of which is adapted from \citet{Vialard2009} and \citet{Begetal2005}.

\begin{lemma} 
\label{delta_u_right}
Let $u:\mathbb{R}\rightarrow\mathfrak{g}$, $t\mapsto u(t)$ be a curve in $\mathfrak{g}$ and $\varepsilon\mapsto u_\varepsilon$ a variation of this curve. Then
\[
\delta g^u_{t,s}:=\left.\frac{d}{d\varepsilon}\right|_{\varepsilon=0}g_{t,s}^{u_ \varepsilon}
=g^u_{t,s}\int_s^t\left(\operatorname{Ad}_{g_{s,r}^u}\delta u(r)\right) dr\in T_{g^u_{t,s}}G.
\]
\end{lemma}
\begin{proof} For all $\varepsilon$ we have
\[
\frac{d}{dt}g^{u_\varepsilon}_{t,s}=u_\varepsilon(t) g^{u_\varepsilon}_{t,s},\quad g_{s,s}^{u_\varepsilon}=e.
\]
Taking the $\varepsilon$-derivative of this equality yields the ODE
\[
\frac{d}{dt}\left(\left.\frac{d}{d\varepsilon}\right|_{\varepsilon=0}g^{u_\varepsilon}_{t,s}\right)=\delta u(t) g^u_{t,s}+u(t) \left(\left.\frac{d}{d\varepsilon}\right|_{\varepsilon=0}g^{u_\varepsilon}_{t,s}\right),
\]
and then, using the notation $\delta g^u_{t,s} :=\left.\frac{d}{d\varepsilon}\right|_{\varepsilon=0}g^{u_\varepsilon}_{t,s}$, we compute
\begin{align*}
\frac{d}{dt}\left(\left(g^u_{t,s}\right)^{-1} \delta g^u_{t,s} \right) 
&=-\left(g^u_{t,s}\right)^{-1} u(t) g^u_{t,s} \left(g^u_{t,s}\right)^{-1} \delta g^u_{t,s} 
+ \left(g^u_{t,s}\right)^{-1} \left( \delta u(t) g^u_{t,s} + u(t) \delta g^u_{t,s} \right) \\
&= g^u_{s,t} \delta u(t) g^u_{t,s} \\
&= \operatorname{Ad}_{g^u_{s,t}} \delta u(t).
\end{align*}
Now we integrate both sides from $s$ to $t$ and multiply by $g^u_{t,s}$ from the left to get
\[
\delta g^u_{t,s}=g^u_{t,s}\int_s^t\left(\operatorname{Ad}_{g_{s,r}^u}\delta u(r)\right) dr
\,,\]
as required.
\end{proof}

{
\paragraph{Notation and definitions for cotangent lifts.}
Already knowing from \eqref{eq-mot-grad-abstr} how the first derivative $DE(u_t)$ of the cost functional is going to look, we want to establish the necessary notation before we proceed with the rest of the calculation. 
\begin{itemize}
\item
The inner product on $V$ provides a way to identify $V$ with its dual. To $I \in V$ one associates the linear form $I^\flat:=\langle I, \,{\cdot}\, \rangle\in V^*$.
\item
Given an action $G$ on $V$, we define the \emph{cotangent lift action} of $G$ on $\pi\in V^\ast$ via
\[
\left\langle g\pi, I \right\rangle =\left\langle \pi, g^{-1}I \right\rangle,\quad\text{for all $I\in V$}.
\]
As mentioned earlier in Remark \ref{left-act-inv}, the inverse in this definition is necessary to make the dual action $G\times V^\ast \to V^\ast$ into a left action. 
\item
Finally we define the \emph{cotangent-lift momentum map} $\diamond: V \times V^\ast \to \mathfrak{g}^\ast$ via
\[
\left\langle I \diamond \pi, u \right\rangle = \left\langle \pi, uI \right\rangle,
\]
where $uI$ is the infinitesimal action of $\mathfrak{g}$ on $V$ defined by $uI = \partial_t \vert_{t=0} g_tI$ for a curve $g_t$ with $g_0=e$ and $\partial_t \vert_{t=0} g_t=u$. The use of the momentum map was motivated in Remark \ref{rem-MomMap}.
\end{itemize}
Now we are ready to calculate the stationarity condition $DE(u_t)=0$.
}

\begin{theorem} 
\label{delta_E_left}
Given a curve $t\mapsto u_t\in\mathfrak{g}$, we have
\begin{equation} 
DE(u_t)=0\;\Longleftrightarrow \;\frac{\delta\ell}{\delta u}(t)=-g^u_t I_0 \diamond g^u_{t,1} \pi
\,,
\label{mommap-def1}
\end{equation} 
or, equivalently
\begin{equation} 
DE(u_t)=0\;\Longleftrightarrow\; \frac{\delta\ell}{\delta u}(t)
=-\frac{1}{\sigma^2}J^0_t\diamond \left(g^u_{t,1} \left(J^0_1-J^1_1\right)^\flat\right),
\label{mommap-def2}
\end{equation} 
where the quantities $\pi$, $J_t^0$, and $J^1_t$ are defined as
\[
\pi: = \frac{1}{ \sigma^2}\left(g^u_1 I_0 - I_1\right)^\flat \in V ^\ast,\quad J^0_t = g^u_tI_0\in V,\quad J^1_t= g_{t,1}^uI_1\in V.
\]
When $G$ acts by isometries, the stationarity condition simplifies to
\[
DE(u_t)=0\;\Longleftrightarrow \;\frac{\delta\ell}{\delta u}(t)
=-\frac{1}{\sigma^2}J^0_t\diamond\left(J^0_t-J^1_t\right)^\flat.
\]
\end{theorem}

\noindent The quantity $J^0_t$ is the template object moved \emph{forward} by $g_t$ until time $t$ and $J^1_t$ is the target object moved \emph{backward} in time from $1$ to $t$.

\begin{proof}
Using the notation $\pi: = \frac{1}{ \sigma^2}(g_1^u I_0 - I_1)^\flat = \frac{1}{\sigma^2}(J^0_1-J^1_1)^\flat\in V ^\ast$, we may calculate
\begin{align*}
\left\langle DE(u),\delta u\right\rangle&=\delta \left(\int_0^1\ell(u(t))dt+\frac{1}{2\sigma^2}\left\|g^u_1 I_0-I_1\right\|^2_V\right)\\
&=\int_0^1\left\langle \frac{\delta \ell}{\delta u}(t),\delta u(t)\right\rangle dt
+\left\langle \pi,\left.\frac{d}{d\varepsilon}\right|_{\varepsilon=0}\left(g^{u_\varepsilon}_1 I_0-I_1 \right)\right\rangle\\
&=\int_0^1\left\langle \frac{\delta \ell}{\delta u}(t),\delta u(t)\right\rangle dt+\left\langle \pi, \delta g^u_1 I_0\right\rangle\\
&=\int_0^1\left\langle \frac{\delta \ell}{\delta u}(t),\delta u(t)\right\rangle dt+\left\langle \pi, \left(g^u_1\int_0^1\left(\operatorname{Ad}_{g_{0,s}^u}\delta u(s)\right)ds\right)I_0\right\rangle\\
&=\int_0^1\left( \left\langle \frac{\delta \ell}{\delta u}(t),\delta u(t) \right\rangle dt +\left\langle \left(g^u_1\right)^{-1}\pi,\left(\operatorname{Ad}_{g_{0,t}^u}\delta u(t)\right)I_0\right\rangle \right)dt \\
&=\int_0^1 \left(\left\langle \frac{\delta \ell}{\delta u}(t),\delta u(t)\right\rangle +\left\langle I_0\diamond\left(g^u_1\right)^{-1}\pi,\operatorname{Ad}_{g_{0,t}^u}\delta u(t)\right\rangle \right) dt\\
&=\int_0^1\left(\left\langle \frac{\delta \ell}{\delta u}(t)+\operatorname{Ad}^*_{g_{0,t}^u} \left(I_0\diamond\left(g^u_1\right)^{-1}\pi\right),\delta u(t)\right\rangle\right)dt,
\end{align*}
which must hold for all variations $\delta u(t)$. Therefore,
\begin{align*}
\frac{\delta \ell}{\delta u}(t)&=-\operatorname{Ad}^*_{g_{0,t}^u} \left(I_0\diamond\left(g^u_1\right)^{-1}\pi\right)\\
&=-\,g^u_t I_0 \diamond g^u_{t,1} \pi\\
&=-\,\frac{1}{\sigma^2}J^0_t\diamond g^u_{t,1} \left(J^0_1-J^1_1\right)^\flat.
\end{align*}
If $G$ acts by isometries, then the action commutes with the flat map and we obtain
\[
\frac{\delta \ell}{\delta u}(t)
=-\frac{1}{\sigma^2}J^0_t\diamond\left(J^0_t-J^1_t\right)^\flat.
\]
{
The last expression involving diamond is the cotangent-lift momentum map $\diamond : V\times V^* \to \mathfrak{g}^*$ associated to the given representation of the Lie group $G$ on the vector space $V$. 
}
\end{proof}
This theorem tells us how to compute the gradient of the cost functional for any data structure and any group action. Just like the cost functional \eqref{eq-abfrm-cost} it is expressed entirely in geometric terms and contains no reference to particular examples such as images. This makes the theorem widely applicable.

\begin{remark}\normalfont
Although the momentum $\frac{\delta \ell}{\delta u}(t)$ at each time depends on $I_0$ and $I_1$, it turns out that $\frac{\delta \ell}{\delta u}(t)$ obeys a dynamical equation that is independent of $I_0$, $I_1$. The equation in question is the \textit{Euler-Poincar\'e equation} on $G$. History and applications of the Euler-Poincar\'e equation can be found in \citet{HoMaRa1998}, \citet{MaRa1999} and \citet{Marsden1983}.
\end{remark}

\begin{lemma}\label{EP}
The momentum $\frac{\delta \ell}{\delta u}(t)$ satisfies
\begin{equation}
\frac{d}{dt} \frac{\delta \ell}{\delta u}(t) = -\operatorname{ad}^\ast_{u_t} \frac{\delta \ell}{\delta u}(t)\,.
\label{EP-eqn-left}
\end{equation}
This is the Euler-Poincar\'e equation on the Lie group $G$ with Lagrangian $\ell: TG/G\simeq\mathfrak{g} \to \mathbb{R}$.
\end{lemma}

\begin{proof}
Because the cotangent-lift momentum map is $\operatorname{Ad}^\ast$-invariant we obtain from Theorem \ref{delta_E_left}
\begin{align*}
\frac{\delta \ell}{\delta u}(t) &= -g^u_t I_0 \diamond g^u_{t,1} \pi\\
&= -\operatorname{Ad}^*_{(g_t^u)^{-1}} \left(I_0\diamond\left(g^u_1\right)^{-1}\pi\right).
\end{align*}
Differentiation of $\operatorname{Ad}^\ast$ follows the rules
\begin{align*}
\partial_t \operatorname{Ad}^\ast_{g_t} \eta &= \operatorname{Ad}^\ast_{g_t} \operatorname{ad}^\ast_{\dot g_t g_t^{-1}} \eta
\,, \\
\partial_t \operatorname{Ad}^\ast_{g_t^{-1}} \eta &= -\operatorname{ad}^\ast_{\dot g_t g_t^{-1}} \operatorname{Ad}^\ast_{g_t} \eta
\,.
\end{align*}
From this we see that
\begin{align*}
\frac{d}{dt} \frac{\delta \ell}{\delta u}(t) &= - \frac{d}{dt} \operatorname{Ad}^*_{(g_t^u)^{-1}} \left(I_0\diamond\left(g^u_1\right)^{-1}\pi\right) \\
&= \operatorname{ad}^\ast_{u_t} \operatorname{Ad}^\ast_{g_t} \left(I_0\diamond\left(g^u_1\right)^{-1}\pi\right) \\
&= -\operatorname{ad}^\ast_{u_t} \frac{\delta \ell}{\delta u}(t)
\,,
\end{align*}
and so the momentum satisfies the Euler-Poincar\'e equation.
\end{proof}

\begin{remark}[EPDiff equation]\normalfont
When $G={\rm Diff}(M)$ the  Euler-Poincar\'e equation is the EPDiff equation for left action of the diffeomorphisms on the manifold $M$,
\begin{equation}
\frac{d}{dt} \frac{\delta \ell}{\delta u}(t) 
=
-\operatorname{ad}^\ast_{u_t} \frac{\delta \ell}{\delta u}(t)
\,.
\label{EPDiff-eqn-left}
\end{equation}
See \citet{HoMa2004} for a detailed treatment of the EPDiff equation and see \citet{YoArMi2009} for interesting discussions of its various usages in computational anatomy.
\end{remark}

\begin{remark}[Dependence of $I_0, I_1$ on the initial momentum]
\normalfont
It might seem counterintuitive that the momentum evolves independently of the objects we are trying to match. However, the objects $I_0, I_1$ do influence the momentum $\frac{\delta \ell}{\delta u}(t)$ in a significant way. Namely, solving the Euler-Poincar\'e equations requires that we know the initial momentum $\frac{\delta \ell}{\delta u}(0)$ and this initial momentum depends on $I_0$, $I_1$ through the formula
\[
\frac{\delta \ell}{\delta u}(0) = -I_0 \diamond (g^u_1)^{-1} \pi
\,.
\]

Alternatively, we might think of it from the viewpoint of the variational principle. Assume that $\ell(u) = \frac{1}{2} \lvert u\rvert^2$ is the squared length of a vector for some inner product $\langle\,\cdot\,,\,\cdot\,\rangle$ on $\mathfrak g$. If we have found a vector field $u_t$ and $g_1$, which minimize
\[
\frac 1 2 \int_0^1 \lvert u\rvert^2 dt + \frac {1}{2\sigma^2} \lVert g_1I_0 - I_1 \rVert^2_V
\,,
\]
then the vector field $u_t$ must also minimize
\[
\int_0^1 \lvert u\rvert^2 dt
\,,
\]
among all vector fields $\widetilde u_t$ whose flows $\widetilde g_t$ coincide with $g_t$ at time $t=1$, i.e., $\widetilde g_1 = g_1$. But this means that $u_t$ must be the velocity vector field of a geodesic $g_t$ in $G$. Here we have implicitly endowed $G$ with a right-invariant Riemannian metric induced by the inner product $\langle\,\cdot\,,\,\cdot\,\rangle$ on $\mathfrak g$. The Euler-Poincar\'e equation (\ref{EPDiff-eqn-left}) is just the geodesic equation on the Lie group $G$ with respect to this Riemannian metric.
\end{remark}

\section{Registration Using the Group of Diffeomorphisms} 
\label{sec-GeomRegDiff}

\subsection{The Setting} 
In computational anatomy the group of deformations $G$ is usually the group of diffeomorphisms of some domain $\Omega \subset \mathbb R^d$. Different types of data used in computational anatomy, such as landmarks, scalar-valued images or vector fields, are deformed by diffeomorphisms via the mathematical operations of pull-back and push-forward. Intuitively this corresponds to embedding your data into the domain $\Omega$, then deforming $\Omega$ by the diffeomorphism and observing how the data is deformed with it. We will go into greater detail about how each of the data types can be registered after reviewing some basic notions about the diffeomorphism group.

\subsubsection*{Diffeomorphism group}
For technical reasons, we need to consider a group of diffeomorphisms associated to a certain Hilbert space of vector fields $\mathcal{H}$. We suppose that $\mathcal{H}$ is a subspace of the space of $C^1$ vector fields vanishing at the boundary and at infinity, and such that there exists a constant $C$ for which
\begin{equation}\label{hypothesis_on_H}
|u|_{1,\infty}\leq C|u|_\mathcal{H},
\end{equation}
where $|\cdot|_\mathcal{H}$ is the inner product norm of the Hilbert space $\mathcal{H}$ and $|\cdot|$ is the norm in $W^{1,\infty}(\Omega)$. Such a Hilbert space defines a unique Kernel $K:\Omega\times\Omega\rightarrow L(\mathbb{R}^d,\mathbb{R}^d)$ such that
\[
\langle u, p \rangle_{L^2} = \left\langle u, \int K(\cdot ,y)p(y) dy \right\rangle_\mathcal{H}.
\]
This also defines a positive, self-adjoint differential operator $L$ (with respect to the $L^2$-inner product) such that $\left\langle u,v\right\rangle_\mathcal{H}=\left\langle u,Lv\right\rangle_{L^2}$.

If $u_t: [0,1] \to \mathcal{H}$ is a time-dependent vector field in $L^1([0,1],\mathcal{H})$, then following \citet{Younes2010} and \citet{Vialard2009}, we can consider the solution $\phi_t$ of the differential equation
\begin{equation}
\label{eq_def_G_H}
\partial_t \phi_t(x) = u_t \circ \phi_t(x), \quad \phi_0(x) = x,
\end{equation}
and the group
\begin{equation}\label{defintion_G_H}
G_{\mathcal{H}} = \left\{ \phi_1\mid \phi_t \textrm{ is solution of \eqref{eq_def_G_H} for some } u_t\in L^1([0,1],\mathcal{H}) \right\}.
\end{equation}
We shall quickly indicate why $G_\mathcal{H}$ is a group, following \citet{Trouve1995}. Let $\phi^u_1$ and $\phi^v_1$ be the flows at time $t=1$ of the vector fields $u_t$ and $v_t$. Let $\tilde u_t:=-u_{1-t}$. Then we have the relation
\[
\phi^{\tilde u}_t\circ\phi^u_1=\phi^u_{1-t}
\,,
\]
since $\phi^v_t\circ\phi^u_1(x)$ and $\phi^u_{1-t}(x)$ are both integral curves of $\tilde u_t$ at $\phi^u_1(x)$. Taking $t=1$, we obtain $(\phi_1^u)^{-1}=\phi_1^{\tilde u}\in G_\mathcal{H}$. To prove that the composition $\phi^u_1\circ\phi^v_1$ is in $G_\mathcal{H}$, we consider the vector field
\[
(u\star v)_t:=\left\{\begin{array}{ll} 2u_{2t},\quad &\text{if $t\leq 1/2$}\\
2v_{2t-1},\quad &\text{if $t>1/2$}
\end{array}, \qquad t\in[0,1].
\right.
\]
In order to compute $\phi^{u\star v}_1$, we first solve the ODE for $t\leq 1/2$. In this case $(u\star v)_t=2u_{2t}=:\bar u_t$, therefore $\phi^{u\star v}_t=\phi^{\bar u}_t=\phi^u_{2t}$. We then consider the case when $t$ becomes larger than $1/2$. In this case $(u\star v)_t=2v_{2t-1}=:\bar v_t$ and from the situation $t\leq 1/2$, we know that at time $t=1/2$ the flow $\phi^{u\star v}_t$ takes the value $\phi^u_1$. Thus, we must have $\phi^{u\star v}_t=\phi^{\bar v}_t\circ\left(\phi_{1/2}^{\bar v}\right)^{-1}\circ\phi^u_1$. Now we observe that $\phi^{\bar v}_t\circ\left(\phi_{1/2}^{\bar v}\right)^{-1}=\phi^v_{2t-1}$, since they are both integral curves of $\bar v$ that coincide at time $t=1/2$. We thus get the formula
\[
\phi^{u\star v}_t=\phi^v_{2t-1}\circ \phi^u_1.
\]
Taking $t=1$, we get $\phi^v_1\circ\phi^u_1=\phi^{u\star v}_1\in G_\mathcal{H}$.

Even though $G_{\mathcal{H}}$ is not precisely a Lie group, it comes close enough for our purposes, with $\mathcal{H}$ acting as a substitute for the Lie algebra. We can use formal analogies with the finite dimensional case to develop applications for computational anatomy. Details about this construction can be found in \citet{Younes2010}, \citet{Trouve1995} and results about the regularity of the diffeomorphisms thus constructed are found in \citet{TrYo2005a} and in \citet{Glaunes2005}.

In the following, when we speak of the group of diffeomorphisms, we will mean the group $G_{\mathcal{H}}$.

\subsection{Example 1: Landmark Matching}\label{ex1}

The simplest kind of objects used in computational anatomy are landmarks. Landmarks are labeled collections $I=(\mathbf{x}^1,\ldots,\mathbf{x}^n)$ of points $\mathbf{x}^i \in \mathbb R^d$. Given two sets $(\mathbf{x}^1,\ldots,\mathbf{x}^n)$, $(\mathbf{y}^1,\ldots,\mathbf{y}^n)$ of landmarks, the landmark matching problem consists of minimizing the energy
\begin{equation}
\label{energy_landmarks}
E(u_t) = \frac 1 2 \int_0^1 \lvert u_t \rvert^2_\mathcal{H} dt + \frac {1}{2\sigma^2} \sum_{i=1}^n \lVert \phi_1(\mathbf{x}^i) - \mathbf{y}^i \rVert^2.
\end{equation}
Our space of deformable objects is $V=(\mathbb R^d)^n$ with the usual inner product
\[
\langle I, J \rangle = \sum_{i=1}^n \mathbf{x}^i \boldsymbol{\cdot} \mathbf{y}^i
\,,
\]
for $I=(\mathbf{x}^1,\ldots,\mathbf{x}^n)$, $J=(\mathbf{y}^1,\ldots,\mathbf{y}^n)$. The action of the diffeomorphism group $G_{\mathcal{H}}$ is by push-forward
\[
\phi\cdot I := \left(\phi(\mathbf{x}^1),\ldots,\phi(\mathbf{x}^n) \right).
\]
The corresponding cotangent-lift action on the dual space $(\mathbb R^{dn})^\ast \cong \mathbb R^{dn}$ is given by
\[
\phi\cdot J^\flat = \left( D \phi (\mathbf{x}^1)^{-T}\mathbf{y}^1, \ldots, D \phi(\mathbf{x}^n)^{-T}\mathbf{y}^n \right)
\]
and the calculation
\begin{align*}
\left\langle I \diamond J^\flat, 
u \right\rangle_{\mathcal{H}^\ast \times \mathcal{H}} 
&= \left\langle J^\flat, u I \right\rangle \\
&= \left\langle (\mathbf{y}^1,\ldots,\mathbf{y}^n), 
(u(\mathbf{x}^1),\ldots, u(\mathbf{x}^n)) \right\rangle \\
&= \sum_{i=1}^n \mathbf{y}^i \boldsymbol{\cdot} 
u(\mathbf{x}^i) \\
&= \left\langle \sum_{i=1}^n \mathbf{y}^i 
\delta_{\mathbf{x}^i}, u \right\rangle_{\mathcal{H}^\ast \times \mathcal{H}}
\end{align*}
yields the diamond operator (momentum map)
\[
(\mathbf{x}^1,\ldots,\mathbf{x}^n) \diamond (\mathbf{y}^1,\ldots,\mathbf{y}^n)^\flat = \sum_{i=1}^n \mathbf{y}^i \delta_{\mathbf{x}^i}
\]
where $\delta_{\mathbf{x}}$ is the delta-distribution defined by $\int f(\mathbf{y}) \delta_\mathbf{x}(\mathbf{y}) d\mathbf{y} = f(\mathbf{x})$ for a test function $f(\mathbf{y})$.

The condition \eqref{mommap-def2} that a minimizing vector field $u_t$ must satisfy is
\[
Lu_t=-\frac{1}{\sigma^2}\sum_{i=1}^nD\phi_{t,1}(\phi_1(\mathbf{x}^i))^{-T} (\phi_1(\mathbf{x}^i)-\mathbf{y}^i)\,\delta_{\phi_t(\mathbf{x}^i)}
\,.
\]
Consequently, the momentum $Lu_t$ is concentrated only on the points $\phi_t(\mathbf{x}^i)$. By using the Green's function $K(\mathbf{x},\mathbf{y})$ corresponding to the differential operator $L$, the minimizing condition above can be rewritten for the velocity $u_t$ as
\[
u_t(\mathbf{x}) = -\frac{1}{\sigma^2} \sum_{i=1}^n K(\mathbf{x}, \phi_t(\mathbf{x}^i))\left[D \phi_{t,1}(\phi_1(\mathbf{x}^i))^{-T}(\phi_1(\mathbf{x}^i) - \mathbf{y}^i)\right].
\]

\subsection{Example 2:  Image Matching} 
\label{subsec-imagematching}

The large deformation diffeomorphic matching framework used in 
\citet{Begetal2005} seeks to match two images $I_0, I_1$ by minimizing
\[
E(u_t) = \frac 1 2 \int_0^1 \lvert u_t \rvert^2_\mathcal{H} dt + \frac {1}{2\sigma^2} \lVert I_0 \circ \phi_1^{-1} - I_1 \rVert^2_{L^2}.
\]
This example has already been discussed in Section \ref{subsec-Motivation}. We review it here by applying the abstract formalism developed above.
In this example the space $V$ of deformable objects consists of real valued functions on $\Omega$. We endow this space with the $L^2$-inner product. The group of deformations is again the group of diffeomorphisms $G_{\mathcal{H}}$, generated by vector fields in $\mathcal{H}$. The action of $G_{\mathcal{H}}$ on $V$ is by push-forward
\[
\phi\cdot I = \phi_\ast I = I \circ \phi^{-1}
\]
for $\phi \in G_{\mathcal{H}}$ and $I \in V$. As we have seen, the dual action reads
\[
\phi\cdot \pi = \lvert \det D\phi^{-1} \rvert\,\left(\pi\circ \phi^{-1}\right)
\]
where $\lvert \det D\phi \rvert$ denotes the absolute value of the determinant of $D \phi$. The diamond map in this example is
\[
I \diamond \pi = -\pi \nabla I
\,.
\]

According to \eqref{mommap-def2}, a minimizing vector field $u_t$ must satisfy the following necessary condition
\begin{equation}
\label{beg_eq}
L u_t = \frac{1}{\sigma^2} \lvert \,\det D\phi_{t,1}^{-1} \rvert (J^0_t - J^1_t)\nabla J^0_t
\end{equation}
where $J^0_t = I_0 \circ \phi^{-1}_{t,0}$, $J^1_t = I_1 \circ \phi^{-1}_{t,1}$, and $\phi_{t,s}$ is the flow of the vector field $u_t$
\[
\partial_t \phi_{t,s} = u_t \circ \phi_{t,s}, \quad \phi_{s,s}(x) = x.
\]
Equation \eqref{beg_eq} was used in \citet{Begetal2005} in devising a gradient descent scheme to computationally find the minimizing vector field.

\subsection{Example 3: Vector Fields} \label{example-VFs}

Diffusion tensor magnetic resonance imaging measures the anisotropic diffusion of water mole\-cules in biological tissues, thus enabling us to quantify the structure of the tissue. The measurement at each voxel is a second order symmetric tensor. It was shown in \citet{Pierpaoli1996} and \citet{Scollan1998} that the alignment of the principal eigenvector of this tensor tends to coincide with the fiber orientation in brain and heart.

The fiber orientation can be described by a vector field $I: \Omega \to \mathbb R^d$ and matching two vector fields can be formulated as minimizing the energy
\begin{equation}
\label{energy_vf}
E(u_t) = \frac 1 2 \int_0^1 \lvert u_t \rvert^2_\mathcal{H} dt + \frac {1}{2\sigma^2} \lVert  D \phi_1\circ I_0\circ\phi_1^{-1} - I_1 \rVert^2_{L^2}.
\end{equation}
In this example the space of deformable objects $V$ is the vector space of vector fields in $\Omega$, the deformation group is the group of diffeomorphisms $G_{\mathcal{H}}$, generated by vector fields in $\mathcal{H}$, and $G_{\mathcal{H}}$ acts on $V$ by push forward
\[
\phi\cdot I = \phi_\ast I = D\phi \circ I \circ \phi^{-1}.
\]
The infinitesimal action of $u \in \mathcal{H}$ on $I \in V$ is given by the negative of the Jacobi-Lie bracket whose components are
\[
(uI)^i = \frac{\partial u^i}{\partial x^j}I^j - \frac{\partial I^i}{\partial x^j}u^j= -[u,I]^i.
\]
The object dual to vector fields with respect to the $L^2$-pairing are one-forms $\pi\in V^\ast = \Omega^1(\Omega)$. The diamond map is given by
\[
I \diamond \pi = - \pounds_I\pi-\operatorname{div}(I)\pi,
\]
where $\pounds_I\pi$ denotes the Lie derivative of the one-form $\pi$ along the vector field $I$. In coordinates, writing $I = I^i \frac{\partial}{\partial x^i}$ and $\pi =\pi_i dx^i$, we can write the diamond map in the form
\[
I \diamond \pi = - \left( \pi_j \frac{\partial I^j}{\partial x^i} + I^j \frac{\partial \pi_i}{\partial x^j} + \pi_i \frac{\partial I^j}{\partial x^j}\right) dx^i.
\]
{
Again, diamond denotes the momentum map $\diamond : V\times V^* \to \mathfrak{g}^*$ for images that are vector fields $I \in V=\mathfrak{X}(\Omega)$ and their duals $I^\flat:=\langle I, \,{\cdot}\, \rangle\in V^*=\mathfrak{X}^*(\Omega)\simeq \Lambda^1( \Omega )\times Dens(\Omega)$, the 1-form densities. 
}

Using these formulas, we can write the necessary condition for a vector field $u_t$ to minimize \eqref{energy_vf} as
\[
L u_t = \left(\pounds_{(\phi_t)_\ast I_0} +\operatorname{div}\left((\phi_t)_*I_0\right)\right) \left( \lvert \det D\phi_{t,1}^{-1} \rvert (\phi_{t,1})_\ast \pi \right),
\]
where
$\pi = \frac{1}{\sigma^2} \left( (\phi_1)_\ast I_0 - I_1 \right)^\flat\in V^*$. Note that because the $\flat$-map does not commute with pull backs and push forwards, i.e.
\[
\phi^\ast (\phi_\ast I)^\flat \neq I^\flat,
\]
this formula cannot be significantly simplified.

\subsection{Diffusion Tensor MRI}
\label{DT-MRI}

Instead of matching only the fiber orientations, we could also match the entire symmetric 2-tensor, as was done in \citet{Alex-etal-2001} and \citet{Cao2006}. In order to do so, we should first explain how a diffusion tensor changes under a diffeomorphism. In analogy to images and vector fields we could use the push forward by the diffeomorphism. If $T$ is a symmetric tensor-field with coordinates $T_{ij}$, i.e.
\[
T(x) = T_{ij} dx^i \otimes dx^j
\]
and $\phi \in G_{\mathcal{H}}$ a diffeomorphism, then the push-forward has the coordinate expression
\begin{equation}
\label{push_forward_action}
\phi_\ast T(x) = T_{ij}(\phi^{-1}(x)) B_k^i(x) B_l^j(x) dx^k \otimes dx^l,
\end{equation}
where and $B_k^i(x)$ is the coordinate matrix of $D\phi^{-1}(x)$.

In \citet{Alex-etal-2001} and \citet{Cao2006} a different action was used. At each point $x\in \Omega \subset \mathbb{R}^d$ the orthonormal principal-axis directions $\mathbf{e}_1(x), \mathbf{e}_2(x), \mathbf{e}_3(x)$ of the tensor $T(x)$ are computed, as well as their corresponding  eigenvalues $\lambda_1(x) 
\geq \lambda_2(x) \geq \lambda_3(x)$.  
Then $T $ can be written as $T = \lambda_1 \mathbf{e}_1 \mathbf{e}_1^T +\lambda_2 \mathbf{e}_2 \mathbf{e}_2^T + \lambda_3 \mathbf{e}_3 \mathbf{e}_3^T $.
The principal axes are each transformed separately as vector fields under the diffeomorphisms as in Section \ref{example-VFs}, then normalized and made orthogonal using the Gram-Schmidt method. The results are given as:
\begin{align*}
\mathbf{\widehat e}_1 &= \frac{\phi_*\mathbf{e}_1}{\| \phi_*\mathbf{e}_1 \|} \,,\\
\mathbf{\widehat e}_2 &= \frac{ \phi_*\mathbf{e}_2 
- \langle \,\mathbf{\widehat e}_1, \phi_*\mathbf{e}_2 \rangle \mathbf{\widehat e}_1}{\| \phi_*\mathbf{e}_2 
- \langle \,\mathbf{\widehat e}_1, \phi_*\mathbf{e}_2 \rangle \mathbf{\widehat e}_1\|} \,,\\
\mathbf{\widehat e}_3 &= \mathbf{\widehat e}_1 \times \mathbf{\widehat e}_2.
\end{align*}
In the above lines, the first principal axis $\mathbf{e}_1$ is pushed forward by $\phi$ to $\mathbf{\widehat e}_1$ parallel to  $\phi_*\mathbf{e}_1$. The second principal axis $\mathbf{e}_2$ is mapped in such a way, that $\mathbf{\widehat e}_1, \mathbf{\widehat e}_2$ span the same plane as $\phi_*\mathbf{e}_1, \phi_*\mathbf{e}_2$ and are orthogonal to each other. The last principal axis is then mapped to be orthogonal to the first two. The transformed tensor is \emph{defined} to be:
\begin{equation}
\label{other_action}
\phi\cdot T = \lambda_1\mathbf{\widehat e}_1 \mathbf{\widehat e}_1^T 
+ \lambda_2 \mathbf{\widehat e}_2 \mathbf{\widehat e}_2^T 
+ \lambda_3\mathbf{\widehat e}_3 \mathbf{\widehat e}_3^T.
\end{equation}
This means, that we transform the principal axis \emph{directions} as described above, but we do \emph{not} change the eigenvalues. The choice of this action is motivated by the particular application. In brain DT-MRI the tensor $T(x)$ describes the diffusivity of water in different directions at a position $x$. The action by diffeomorphisms describes a \emph{macroscopic} deformation of the brain, such as a change of orientation, a growing tumor or a trauma. However, the diffusivity of water is governed by the \emph{microscopic} structure of tissue, which remains unchanged under a macroscopic transformation. Therefore, one is looking for a way to transform the tensor, while keeping its eigenvalues (the principal diffusivities) unchanged.

It can be shown that $T \mapsto \phi\cdot T $ is a left action of $\operatorname{Diff}( \Omega)$ on the vector space of symmetric two-tensors.
Both of these approaches to Diffusion Tensor MRI given by the actions \eqref{push_forward_action} and \eqref{other_action} have the structure of a Lie group action and thus they may both be cast into our momentum-map framework. We leave it to future work to study the different momentum maps that arise for each of these actions and the implications that they have for matching of DT-MRIs.

\section{Registration using Semidirect Products} 
\label{sec-GeomReg-SDP}

The examples in the previous section have shown that the abstract formulation of diffeomorphic image registration using the diamond operation $(\,\diamond\,)$ provides a mathematical framework that allows us to adapt easily to accommodate different data structures. A second advantage of this framework is the ability to perform matching using different groups. The images encountered in computational anatomy may contain information on different length scales. Two images can vary in their large scale structure as well as in the fine details. In matching such images, it might be of advantage to have two groups at our disposal, one to match the large scale behavior and the other one to deal with the fine details. This is made possible in our framework by using the concept of a semidirect product, which we will review below and then apply in examples.

\subsection{Semidirect Product of Groups}

Consider a Lie group $H$ acting on $K$ from the left by homomorphisms.
\[
\left(h, k\right) \in H \times K \mapsto  h \cdot k \in K,
\]
that is,
\begin{align*}
h_1 \cdot (h_2 \cdot k) &= (h_1h_2) \cdot k & &\textrm{left group action} \\
h \cdot (k_1 k_2) &= (h \cdot k_1)(h \cdot k_2) & & \textrm{action by group homomorphisms.}
\end{align*}

We can then form the semidirect product group $G = H \,\circledS\, K$. The group multiplication in $G$ is given by
\begin{equation}
g_1 g_2 = \left(h_1, k_1\right) \left(h_2, k_2 \right) = \left(h_1 h_2, k_1 \left(h_1 \cdot k_2\right) \right)
\end{equation}
and the inverse of $(h,k)$ is $(h,k)^{-1} = (h^{-1}, h^{-1} \cdot k^{-1})$. The Lie algebra $\mathfrak{g}$ is the semidirect product $\mathfrak{g} = \mathfrak{h}\, \circledS \,\mathfrak{k}$ of the Lie algebras of $H$ and $K$. The tangent actions on $G$ are given by
\begin{align}
\label{sdp_tan_mul}
(\dot h_1, \dot k_1)(h_2, k_2) &= \left(\dot h_1 h_2, \dot k_1 (h_1 \cdot k_2) + k_1 (\dot h_1 \cdot k_2) \right), \\
(h_1, k_1)(\dot h_2, \dot k_2) &= \left(h_1 \dot h_2, k_1 \cdot (h_1 \dot k_2) \right),
\end{align}
and the right-trivialization of the tangent bundle is given by
\[
\dot g g^{-1} = (\dot h, \dot k)(h^{-1}, h^{-1} \cdot k^{-1})
= \left(\dot h h^{-1}, \dot k k^{-1} + k (\dot h h^{-1} \cdot k^{-1}) \right)
.\]
The next lemma provides formulas for the adjoint and coadjoint actions of $H \,\circledS\, K$ on itself and its Lie algebra.

\begin{lemma}
\label{ad_sdp_left}{\bf [Adjoint and coadjoint actions]}$\,$

We have the following formulas for the adjoint and coadjoint actions
\begin{align}
\label{Ad_sdp_left}
\operatorname{Ad}_{(h,k)} (w,v) &= \left( \operatorname{Ad}_h v, \operatorname{Ad}_k (h \cdot w) + k (\operatorname{Ad}_h v \cdot k^{-1}) \right) \\
\label{coAd_sdp_left}
\operatorname{Ad}^\ast_{(h,k)} (\mu, \nu) &= \left( \operatorname{Ad}^\ast_{h} (\mu + \mathbf{J}(k^{-1} \nu)), h^{-1} \cdot \operatorname{Ad}^\ast_{k} \nu \right) \\
\label{ad_sdp}
\operatorname{ad}_{(v_1,w_1)} (v_2,w_2) &= \left( \operatorname{ad}_{v_1} v_2, \operatorname{ad}_{w_1} w_2+v_1 \cdot w_2 - v_2 \cdot w_1 \right) \\
\label{coad_sdp}
\operatorname{ad}^\ast_{(v_1, w_1)} (\mu, \nu) &= \left( \operatorname{ad}^\ast_{v_1} \mu - w_1 \diamond \nu, \operatorname{ad}^\ast_{w_1} \nu - v_1 \cdot \nu \right),
\end{align}
where $\mathbf{J}: T^\ast K \to \mathfrak{h}^\ast$ is the cotangent lift momentum map associated to the action of $H$ on $K$
\[
\left\langle \mathbf{J}(\alpha_k), v \right\rangle = \left\langle \alpha_k, v \cdot k \right\rangle
\,,
\]
and $\diamond: \mathfrak{k} \times \mathfrak{k}^\ast \to \mathfrak{h}^\ast$ is the cotangent lift momentum map associated to the induced representation of $H$ on $\mathfrak{k}$
\[
\left\langle w \diamond\nu, v \right\rangle := \left\langle \nu, v \cdot w \right\rangle.
\]
The action $(v, w) \in \mathfrak{h} \times \mathfrak{k} \mapsto v \cdot w \in \mathfrak{k}$ is defined as $v \cdot w = \partial_t \vert_{t=0} (h(t) \cdot w)$ for a curve $h(t)$ with $h(0)=e$ and $\partial_t \vert_{t=0} h(t) = v$.
\end{lemma}

\begin{proof}
For the adjoint action (Ad) of the group on its Lie algebra, we simply perform the multiplications
\begin{align*}
\operatorname{Ad}_{(h,k)} (v,w)
&= (h,k)(v,w)\left(h^{-1}, h^{-1} \cdot k^{-1} \right) \\
&= \left(hv, k(h \cdot w) \right) \left(h^{-1}, h^{-1} \cdot k^{-1} \right) \\
&= \left(hvh^{-1}, k(h \cdot w)k^{-1} + k(hvh^{-1} \cdot k^{-1}) \right) \\
&= \left(\operatorname{Ad}_{h} v, \operatorname{Ad}_{k} (h\cdot w) + k (\operatorname{Ad}_{h} v \cdot k^{-1}) \right)
\end{align*}
and for the coadjoint action  (Ad$^*$) on the dual Lie algebra, we pair with $(a,b) \in \mathfrak{h} \,\circledS \,\mathfrak{k}$ to define
\begin{align*}
\left\langle \operatorname{Ad}^\ast_{(h,k)} (\mu, \nu), (a, b) \right\rangle
&= \left\langle (\mu, \nu), \operatorname{Ad}_{(h,k)}(a,b) \right\rangle \\
&= \left\langle \mu, \operatorname{Ad}_h a \right\rangle 
+ \left\langle \nu, \operatorname{Ad}_k (h \cdot b) + k (\operatorname{Ad}_h a \cdot k^{-1}) \right\rangle \\
&= \left\langle \operatorname{Ad}^\ast_h \mu, a \right\rangle 
+ \left\langle h^{-1} \cdot \operatorname{Ad}^\ast_k \nu, b \right\rangle 
+ \left\langle k^{-1} \nu, \operatorname{Ad}_h a \cdot k^{-1} \right\rangle \\
&= \left\langle \operatorname{Ad}^\ast_h \mu, a \right\rangle 
+ \left\langle h^{-1} \cdot \operatorname{Ad}^\ast_k \nu, b \right\rangle 
+ \left\langle \operatorname{Ad}^\ast_h \left(\mathbf{J}( k^{-1} \nu)\right), a \right\rangle \\
&= \left\langle \left( \operatorname{Ad}^\ast_{h} (\mu + \mathbf{J}(k^{-1} \nu)), h^{-1} \cdot \operatorname{Ad}^\ast_{k} \nu \right), (a, b) \right\rangle.
\end{align*}
For the next identity we differentiate \eqref{Ad_sdp_left} and remark that because of $h \cdot e = e$ we get $v \cdot e = 0$. Thus, the adjoint action (ad) of the Lie algebra on itself is given by
\begin{align*}
\operatorname{ad}_{(v_1,w_1)} (v_2,w_2) 
&= \left(\operatorname{ad}_{v_1} v_2, \operatorname{ad}_{w_1} w_2 + v_1 \cdot w_2 + w_1 (v_2 \cdot e) + \operatorname{ad}_{v_1} v_2 \cdot e - v_2 \cdot w_1 \right)\\
&= \left(\operatorname{ad}_{v_1} v_2, \operatorname{ad}_{w_1} w_2 + v_1 \cdot w_2 - v_2 \cdot w_1 \right).
\end{align*}
For the coadjoint action (ad$^*$) of the Lie algebra on its dual, we pair again with $(a, b) \in \mathfrak{h}\, \circledS\, \mathfrak{k}$ to see that
\begin{align*}
\left\langle \operatorname{ad}^\ast_{(v_1, w_1)}(\mu,\nu), (a,b) \right\rangle
&= \left\langle (\mu,\nu), \left(\operatorname{ad}_{v_1} a,
\operatorname{ad}_{w_1} b+ v_1 \cdot b - a \cdot w_1 \right) \right\rangle \\
&= \left\langle \operatorname{ad}^\ast_{v_1} \mu, a \right\rangle
+ \left\langle  \operatorname{ad}^\ast_{w_1} \nu-v_1 \cdot \nu, b \right\rangle
- \left\langle w_1 \diamond \nu, a \right\rangle \\
&= \left\langle \left( \operatorname{ad}^\ast_{v_1} \mu - w_1 \diamond \nu, 
\operatorname{ad}^\ast_{w_1} \nu - v_1 \cdot \nu \right), (a, b) \right\rangle
\end{align*}
as stated in the lemma.
\end{proof}

If $G = H \,\circledS\,K$, the equation
\[
\partial_t g^u_{t,s}=u_tg^u_{t,s},\quad g_{s,s}^u=e.
\]
can be written as (see \eqref{sdp_tan_mul})
\[
\partial_t \left(h^u_{t,s}, k^u_{t,s}\right) = \left(v_t h^u_{t,s}, w_t k^u_{t,s} + 
v_t \cdot k^u_{t,s} \right), \quad h^u_{s,s}=e, k^u_{s,s}=e,
\]
where $u_t= (v_t,w_t) \in   \mathfrak{h}\, \circledS\, \mathfrak{k}= 
\mathfrak{g}$ and   $g^u_{t,s} = (h^u_{t,s}, k^u_{t,s}) \in H \,\circledS\, K$.
Thus $h^u_{t,s}$ and $k^u_{t,s}$ satisfy the equations
\begin{equation}
\label{sdp_rel_k_w}
\partial_t h^u_{t,s} = v _t h^u_{t,s}, \qquad \partial_t k^u_{t,s} = w_t k^u_{t,s} + v_t \cdot k^u_{t,s}.
\end{equation}

This means that $h^u_{t,s}$ is the flow of the vector field $v_t$, but this is not true for $k^u_{t,s}$ and the vector field $w_t$. The corresponding relation for $k^u_{t,s}$ is a direct consequence of the noncommutativity of the semidirect product. After reviewing these facts about the semidirect product, we will apply them to form the semidirect product of two diffeomorphism groups and use this product to perform image registration. This is done in the next section.

\subsection{Image Matching with Semidirect Product Groups} 

Given a space $V$ of deformable objects, assume that two groups $H$, 
$K$ of deformations act on $V$ from the left. We imagine $H$ to contain 
large-scale deformations and $K$ to contain small-scale deformations. 
Since a deformation that captures small structures is also able to capture 
large-scale ones, we will assume that $H$ is a subgroup of $K$, denoted 
by $H \leq K$.

Let us determine the action by group isomorphisms of $H $ on $K $ 
subject to the following two conditions:
\begin{itemize}
\item The formula 
\begin{equation}
\label{abstract_sdp_action}
(h,k)I := khI
\end{equation} 
defines an $H \,\circledS\,K $-action
on $V$. Thus $h $ deforms $I $ first on a large scale and then the details
are captured on a small scale by $k$.
\item The $H \,\circledS\,K $ action is effective. If the action is a 
representation, this means that it is faithful. This condition requires that if
$(h,k)I = I $ for all $I \in V $, then $(h,k) $ is the identity.
\end{itemize}
The first condition implies $(h_1,k_1)(h_2,k_2)I = (h_1 h_2, k_1 
(h_1 \cdot k_2))I$ for all $h_1, h_2 \in H $, $k_1, k_2 \in K $, and $I \in V $.
Therefore, $k_1 h_1 k_2 h_2 I = k_1 (h_1 \cdot k_2) h_1 h_2 I$ for all 
$I \in V $ which, by the second condition, yields $k_1 h_1 k_2 = 
k_1 (h_1 \cdot k_2) h_1$, that is, the action is necessarily given by
conjugation $(h_1 \cdot k_2) =h_1 k_2  h_1^{-1}$. In this sense the action 
by conjugation appears naturally.

Because of the form of the action on $V $, the momentum map
of the cotangent lifted action of $H \,\circledS\,K $ on $V\times V ^\ast$ has the expression
\begin{equation}
\label{diamond_sdp}
I \diamond \pi = (I \diamond_1 \pi, I \diamond_2 \pi) \in \mathfrak{h}^\ast 
\times \mathfrak{k}^\ast \cong ( \mathfrak{h} \,\circledS\,\mathfrak{k})^\ast,
\end{equation}
where, $I \in V $, $\pi\in V ^\ast$, and $I \diamond_1 \pi$ and 
$I \diamond_2 \pi$ denote the cotangent lift momentum maps of the $H $ 
and $K $-actions on $V$, respectively. 

Since the $H $-momentum map is obtained from the $K $-momentum
map by restriction, we have
\begin{equation}
\label{restriction}
\iota^\ast(I \diamond_2 \pi) = I \diamond_1 \pi,
\end{equation}
where $I \in V$, $\pi\in V ^\ast$, $\iota: \mathfrak{h} \hookrightarrow 
\mathfrak{k}$ is the inclusion, and $\iota^\ast: \mathfrak{k}^\ast \rightarrow 
\mathfrak{h} ^{\ast}$ is its dual. 

The matching problem using a semidirect product is to minimize the energy
\begin{equation}
\label{action_left_sdp}
E(v_t, w_t) = \int_0^1 \ell(v_t, w_t) dt + \frac {1}{2\sigma^2} \left\lVert k_1 h_1 I_0-I_1 \right\rVert^2_V, 
\end{equation}
where $(h_1, k_1)$ are related to $(v_t,w_t)$ by
\begin{equation}
\label{diff_equ_h_k}
\left.
\begin{aligned}
\partial_t h_t &= v_t h^u_{t} & \quad &h_0 &= e \\
\partial_t k_t &= (v_t + w_t) k^u_{t} - k^u_{t} v_t &\quad & k_0 &= e
\end{aligned}
\right\}.
\end{equation}
The last equation is obtained by specializing \eqref{sdp_rel_k_w} for $s=0 $ and the action equal to conjugation.

\begin{theorem}
\label{delta_e_sdp}
Given a curve $t \mapsto (v_t, w_t) \in \mathfrak{h} \,\circledS\, 
\mathfrak{k}$ the stationarity condition $DE(v_t, w_t)=0$ for the action 
\eqref{action_left_sdp} is equivalent to
\[
\frac{\delta \ell}{\delta v}(t) = - \widetilde{g}_t I_0 \diamond_1  
\widetilde{g}_{t,1} \pi, \qquad 
\frac{\delta \ell}{\delta w}(t) = - \widetilde{g}_t I_0 \diamond_2 
\widetilde{g}_{t,1} \pi,
\]
where $\pi = \frac{1}{\sigma^2}(\widetilde{g}_1I_0 - I_1)^\flat$ and $
\widetilde{g}_t \in K$ is the solution of the equation
\begin{equation}
\label{g_tilde_left}
\partial_t \widetilde{g}_t = (v_t+w_t) \widetilde{g}_t, \qquad 
\widetilde{g}_0 = e.
\end{equation}
\end{theorem}
\begin{proof}
Let $g_t=(h_t,k_t)\in H \,\circledS\, K$ be the solution of the equation 
$\partial_t g _t = u _tg _t$, $g_0 = e $, where $u_t=(v_t,w_t) \in \mathfrak{h}\,\circledS\,\mathfrak{k}$ and define $\widetilde{g}_t := k_th_t \in K$.  By  Theorem \ref{delta_E_left} and \eqref{diamond_sdp} we get
\[
\frac{\delta \ell}{\delta v}= - g^u_t I_0  \diamond_1  g^u_{t,1} \pi, \qquad  \frac{\delta \ell}{\delta w} 
=  - g^u_t I_0  \diamond_2 g^u_{t,1} \pi .
\]
Since $g_t I_0 = k_t h_t I_0 = \widetilde{g}_t I_0$ by the definition of the
$H \,\circledS\,K $-action on $V $, this yields
\begin{align*}
\frac{\delta \ell}{\delta v}(t) = - \widetilde{g}_t I_0 \diamond_1  \widetilde{g}_{t,1} \pi, \qquad
\frac{\delta \ell}{\delta w}(t) = - \widetilde{g}_t I_0 \diamond_2  \widetilde{g}_{t,1} \pi.
\end{align*}
It remains to show equation \eqref{g_tilde_left}. By \eqref{diff_equ_h_k}
we have 
\begin{align*}
\partial_t \widetilde{g}_t &= (\partial_t k_t) h_t + k_t (\partial_t h_t)
= \left(v _t + w _t \right)k _t h _t - k _t v _t h _t + k _t v _th _t
= \left(v _t + w _t \right)\widetilde{g}_t
\,.
\end{align*}
We have $\widetilde{g}_0 = k_0h_0 = e $.
\end{proof}

This theorem shows that when matching with two groups, the momentum $\frac{\delta \ell}{\delta v}(t)$ contains no more information than $\frac{\delta \ell}{\delta w}(t)$, since we have
\[
\frac{\delta \ell}{\delta v}(t) = \left.\frac{\delta \ell}{\delta w}(t)\right\vert_\mathfrak{h}
\]
by \eqref{restriction}.
Nonetheless, this case \emph{differs} from matching with only one group, since the Euler-Poincar\'e equation for the semidirect product reads
\begin{equation*} 
\begin{split}
\frac{d}{dt} \frac{\delta \ell}{\delta v}(t) &= - \operatorname{ad}^\ast_{v_t} \frac{\delta \ell}{\delta v}(t) + w_t \diamond \frac{\delta \ell}{\delta w}(t) \\
\frac{d}{dt} \frac{\delta \ell}{\delta w}(t) &= -\operatorname{ad}^\ast_{w_t} \frac{\delta \ell}{\delta w}(t) + v_t \cdot \frac{\delta \ell}{\delta w}(t)
\end{split} \end{equation*}
which incorporates the actions of both groups and is genuinely different from the Euler-Poincar\'e equation for a single group, which is
\[
\frac{d}{dt} \frac{\delta \ell}{\delta u}(t) = -\operatorname{ad}^\ast_{u_t} \frac{\delta \ell}{\delta u}(t).
\]

\subsection{Example: Semidirect Product Image Matching with Two Kernels} 

One way of introducing a length scale in image matching is to choose an appropriate  kernel for the cost of the $H$-action. If we were to choose for example $Lu= u - \alpha^2 \Delta u$ to be the differential operator associated to the $H^1$-norm on $H$, then the corresponding kernel would be $K(x,y) = e^{(-\lvert x-y\rvert/\alpha)}$ where $\alpha$ is a length scale; that is, a filter width. A popular alternative choice in image registration is the smoother Gaussian kernel $K(x,y) = e^{(-\lvert x-y\rvert^2/\alpha^2) }$. Increasing the value of $\alpha$ increases the cost of forming gradients, or curvature, and thus inhibits nearby particles from being deformed differently, while allowing large-scale deformations of the image to occur. Sufficiently decreasing the value of $\alpha$ on the other hand would allow fine adjustments in the image to be made without requiring much energy cost for the velocity vector field.

Recall the setting of the example of image matching in Section \ref{subsec-imagematching}. When matching two images $I_0, I_1 \in V : = \mathcal{F}( \Omega)$ with one kernel, the optimizing vector field $u_t$ satisfies
\[
u_t = \frac{1}{\sigma^2} K \ast \left( \lvert \det D\phi_{t,1}^{-1} \rvert \nabla J^0_t (J^0_t - J^1_t) \right),
\]
where $K \ast f = \int K(\cdot ,y)f(y)dy$ denotes convolution with the kernel of the operator $L $; see \eqref{eq-mot-grad}.

A natural approach for distinguishing between multiple length scales would be to use instead the sum of two kernels
\[
\widetilde{u}_t = \frac{1}{\sigma^2} \left(K_{\alpha_1} + K_{\alpha_2}\right) \ast \left( \lvert \det D\widetilde{\phi}_{t,1}^{-1} \rvert \nabla \widetilde{J}^0_t (\widetilde{J}^0_t - \widetilde{J}^1_t) \right),
\]
with two length scales $\alpha_1$ and $\alpha_2$. We will show how this approach can be given a geometrical interpretation.

Given two kernel $K_{ \alpha_1} $ and $K_{ \alpha _2} $ that correspond 
to the two length scales $\alpha _1> \alpha_2 $, we use the diagonal Lagrangian $\ell(v, w) = \frac{1}{2} \lvert v\rvert^2_{\alpha_1} + \frac{1}{2} 
\lvert w\rvert^2_{\alpha_2}$ to measure the energy of the joint velocity vector $(v,w)$. The norm $| \cdot |_{ \alpha_i}$ is associated to the inner 
product coming from the kernel $K_{ \alpha_i}$, $i=1,2$. We assume that the associated Hilbert spaces $\mathcal{H}_{\alpha_1}\subset \mathcal{H}_{\alpha_2}$ verify the hypothesis \eqref{hypothesis_on_H}. Let 
$G_{ \alpha_1} \subset G_{ \alpha_2} $ be the groups associated to $\mathcal{H}_{\alpha_1}$, $\mathcal{H}_{\alpha_2}$ via \eqref{defintion_G_H}. The element $(\psi, \eta) \in 
G_{ \alpha_1} \,\circledS\, G_{ \alpha_2}$ acts on $V = \mathcal{F}
( \Omega) $ by the action \eqref{abstract_sdp_action}; that is, 
\[
(\psi, \eta) \cdot I : = ( \eta \circ \psi) \cdot I = I \circ ( \eta\circ \psi) ^{-1}
= I \circ \psi^{-1} \circ \eta^{-1}.
\]
The matching problem with the semidirect product group 
$G_{\alpha_1} \,\circledS\, G_{\alpha_2}$ is to minimize the energy
\[
E(v_t, w_t) = \frac{1}{2} \int_0^1  \lvert v_t\rvert^2_{\alpha_1} + \lvert w_t\rvert^2_{\alpha_2} dt + \frac {1}{2\sigma^2} \left\lVert I_0\circ \psi_1^{-1}
\circ \eta_1 ^{-1} -I_1 \right\rVert^2_{L^2}.
\]
By Theorem \ref{delta_e_sdp}, the energy is minimal if
\[
v_t = K_{\alpha_1} \ast \left( - \widetilde{\phi}_t I_0 \diamond_1 \widetilde{\phi}_{t,1} \pi \right), \qquad 
w_t = K_{\alpha_2} \ast \left( - \widetilde{\phi}_t I_0 \diamond_2  \widetilde{\phi}_{t,1} \pi \right)
\]
and
\[
\partial_t \widetilde{\phi}_t = (v_t+w_t) \circ \widetilde{\phi}_t, \qquad \phi_0 = id.
\]
The example of single kernel image matching in Section \ref{subsec-imagematching} showed us that
\[
- \widetilde{\phi}_t I_0 \diamond  \widetilde{\phi}_{t,1} \pi 
= \frac{1}{\sigma^2} \lvert\det D\widetilde{\phi}_{t,1}^{-1} \rvert \nabla \widetilde{J}^0_t (\widetilde{J}^0_t - \widetilde{J}^1_t),
\]
with $\widetilde{J}^0_t = I_0 \circ \widetilde{\phi}^{-1}_{t,0}$, $\widetilde{J}^1_t = I_1 \circ \widetilde{\phi}^{-1}_{t,1}$. By denoting 
$\widetilde{u}_t := v_t + w_t$ the velocity vector field of $\widetilde{\phi}_t$, we see that
\begin{equation}
\label{kernel-sum}
\widetilde{u}_t = \frac{1}{\sigma^2} \left(K_{\alpha_1} + K_{\alpha_2}\right) \ast \left( \lvert \det D\widetilde{\phi}_{t,1}^{-1} \rvert \nabla \widetilde{J}^0_t (\widetilde{J}^0_t - \widetilde{J}^1_t) \right).
\end{equation}
{
This computation proves the following theorem.
\begin{theorem}
Matching images with the sum of two kernels corresponds to using a semidirect product of diffeomorphism groups. 
\end{theorem}
\begin{remark} {\rm
This theorem provides a geometrical interpretation for an approach that might have been suggested intuitively and turns out to be very effective.
The sum-of-kernels strategy for registration was recently applied successfully in \cite{RiViWoHoRu2010} for measurement of the atrophy of tissues in the hippocampus due to Alzheimer's disease. }
\end{remark}
}

\section{Symmetric Formulations of Image Registration}
\label{sec-symm}

The cost functional \eqref{eq-abfrm-cost} is not the only choice possible in the large diffeomorphism matching framework. Other cost functionals have been proposed in the literature, which make the registration problem symmetric. A consequence of the choice \eqref{eq-abfrm-cost} is that it matters, whether we choose to register $I_0$ to $I_1$ or vice versa. In some applications it may be useful to distinguish conceptually between $I_0$ and $I_1$. For example, this distinction may be appropriate when the template $I_0$ is available in a higher resolution. However, in other cases one may prefer a symmetric cost functional, instead of \eqref{eq-abfrm-cost}. Such symmetric cost functionals have been proposed in \citet{Beg2007}, \citet{Avants2008} and \citet{ZaHaNi2009}. We will show how they can be analyzed geometrically, much as we did for the cost functional \eqref{eq-abfrm-cost} in Section~\ref{sec-GeomReg}.

\begin{example}
\normalfont
The approach described in \citet{Avants2008} and \citet{Beg2007} can be abstractly described in terms of the following cost functional
\[
E(u_t) = \int_0^1 \ell(u_t) dt + \frac{1}{2\sigma^2} \lVert g_{\frac{1}{2}} I_0 - g_{\frac{1}{2},1} I_1 \rVert^2_V
\]
where $g_{t,s}$ is the flow of $u_t$. Since we now evaluate the inexactness of the matching in the midpoint $t=\frac{1}{2}$ of the interval, this choice of the cost functional leads to a symmetric formulation of LDM. A calculation similar to that in the proof of Theorem \ref{delta_E_left} may be performed with $\pi = \frac{1}{\sigma^2}(g_{\frac{1}{2}} I_0 - g_{\frac{1}{2},1} I_1)^\flat$
\begin{align*}
\left\langle D E(u_t), \delta u_t \right\rangle &= \int_0^1 \left\langle \frac{\delta \ell}{\delta u}(t), \delta u(t) \right\rangle dt
+ \langle \pi, \delta g_{\frac{1}{2}} I_0 - \delta g_{\frac{1}{2},1} I_1 \rangle \\
&= \int_0^1 \left\langle \frac{\delta \ell}{\delta u}(t), \delta u(t) \right\rangle dt \\
&\phantom{=}+ \left\langle \pi, g_{\frac{1}{2}} \left( \int_0^{\frac{1}{2}} \operatorname{Ad}_{g^{-1}_t} \delta u(t) dt \right) I_0
- g_{\frac{1}{2}, 1} \left( \int_1^{\frac{1}{2}} \operatorname{Ad}_{g_{1,t}} \delta u(t) dt \right) I_1 \right\rangle \\
&= \int_0^1 \left\langle \frac{\delta \ell}{\delta u}(t), \delta u(t) \right\rangle dt \\
&\phantom{=}+ \int_0^{\frac{1}{2}} \left\langle I_0 \diamond g^{-1}_{\frac{1}{2}} \pi, \operatorname{Ad}_{g^{-1}_t} \delta u(t) \right\rangle dt
+ \int_{\frac{1}{2}}^1 \left\langle I_1 \diamond g_{1,\frac{1}{2}} \pi, \operatorname{Ad}_{g_{1,t}} \delta u(t) \right\rangle dt \\
&= \int_0^{\frac{1}{2}} \left\langle \frac{\delta \ell}{\delta u}(t) + g_t I_0 \diamond g_{t, \frac{1}{2}} \pi, \delta u(t) \right\rangle dt + \int_{\frac{1}{2}}^1 \left\langle \frac{\delta \ell}{\delta u}(t) + g_{t,1} I_1 \diamond g_{t, \frac{1}{2}} \pi, \delta u(t) \right\rangle dt
\end{align*}
This calculation shows that a minimizing vector field must satisfy
\begin{align*}
\frac{\delta \ell}{\delta u}(t) &=- g_t I_0 \diamond g_{t, \frac{1}{2}} \pi, \quad t \in[0,1/2] \\
\frac{\delta \ell}{\delta u}(t) &= -g_{t,1} I_1 \diamond g_{t, \frac{1}{2}} \pi, \quad t \in [1/2, 1] \\
\pi &= \frac{1}{\sigma^2}(g_{\frac{1}{2}} I_0 - g_{\frac{1}{2},1} I_1)^\flat
\end{align*}
This momentum map is very similar to that of Theorem \ref{eq-abfrm-cost}, except now there is a \emph{discontinuity} at time $t=1/2$.
\end{example}

\begin{example}
\normalfont
Another approach to symmetrize the registration problem was considered in \citet{Beg2007} via the cost functional
\[
E(u_t) = \int_0^1 \ell(u_t) dt + \frac{1}{2\sigma^2} \int_0^1 \lVert g_{t} I_0 - g_{t,1} I_1 \rVert^2_V dt
.\]
Instead of minimizing the matching error at some chosen time (e.g., $t=0$) for the classical LDM or $t=\frac{1}{2}$ as in the previous example, this approach  averages the error over the entire time interval. Upon using the notation $\pi_t = \frac{1}{\sigma^2}(g_{t} I_0 - g_{t,1} I_1)^\flat$ we can again calculate the derivative of $E(u_t)$
\begin{align*}
\left\langle DE(u_t), \delta u_t \right\rangle &= \int_0^1 \left\langle \frac{\delta \ell}{\delta u}(t), \delta u(t) \right\rangle dt
+ \int_0^1 \langle \pi_r, \delta g_{r} I_0 - \delta g_{r,1} I_1 \rangle dr \\
&= \int_0^1 \left\langle \frac{\delta \ell}{\delta u}(t), \delta u(t) \right\rangle dt \\
&\phantom{=}+ \int_0^1 \left\langle \pi_r, g_{r} \left( \int_0^{r} \operatorname{Ad}_{g^{-1}_t} \delta u(t) dt \right) I_0
- g_{r, 1} \left( \int_1^{r} \operatorname{Ad}_{g_{1,t}} \delta u(t) dt \right) I_1 \right\rangle \\
&= \int_0^1 \left\langle \frac{\delta \ell}{\delta u}(t), \delta u(t) \right\rangle dt \\
&\phantom{=}+ \int_0^1 \int _0^r \langle g_tI_0 \diamond g_{t,r} \pi_r, \delta u(t) \rangle dt dr
+ \int_0^1 \int _r^1 \langle g_{t,1}I_1 \diamond g_{t,r} \pi_r, \delta u(t) \rangle dt dr \\
&= \int_0^1 \left\langle \frac{\delta \ell}{\delta u}(t), \delta u(t) \right\rangle dt \\
&\phantom{=}+ \int_0^1 \int_1^t \langle g_tI_0 \diamond g_{t,r} \pi_r, \delta u(t) \rangle dr dt
+ \int_0^1 \int_0^t \langle g_{t,1}I_1 \diamond g_{t,r} \pi_r, \delta u(t) \rangle dr dt \\
&= \int_0^1 \left\langle \frac{\delta \ell}{\delta u}(t)
+ \int_0^1 \left( g_{t,1}I_1 \mathbf{1}_{[0,t]}(r) + g_t I_0 \mathbf{1}_{[t,1]}(r) \right) \diamond g_{t,r} \pi_r dr, \delta u(t) \right\rangle dt
\,.
\end{align*}
This calculation yields the following necessary conditions for the minimizing vector field
\begin{align*}
\frac{\delta \ell}{\delta u}(t)&= - \int_0^1 \left( g_{t,1}I_1 \mathbf{1}_{[0,t]}(r) + g_t I_0 \mathbf{1}_{[t,1]}(r) \right) \diamond g_{t,r} \pi_r dr 
\,,\\
\pi_t &= \frac{1}{\sigma^2}(g_{t} I_0 - g_{t,1} I_1)^\flat
\,.
\end{align*}
Here, $\mathbf{1}_{[0,t]}(r)$ is the indicator function of the interval $[0,t]$, i.e. $\mathbf{1}_{[0,t]}(r)=1$ for $r \in [0,t]$ and 0 otherwise. The momentum map in this case involves an average over time.
\end{example}

\begin{example}
\normalfont
A third approach to symmetric registration was proposed in \citet{ZaHaNi2009}. They suggested that inexactness should be allowed in both the initial and final images, by choosing the cost functional
\[
E(u_t, I) = \int_0^1 \ell(u_t) dt + \frac{1}{2\sigma^2} \lVert I - I_0 \rVert^2_V + \frac{1}{2\sigma^2} \lVert g_1 I - I_1 \rVert^2_V.
\]
This cost functional treats $I \in V$ as an additional free variable. Intuitively, this approach means that we are looking for an energy minimal path such that both the starting and the ending points match $I_0$ and $I_1$ as well as possible. Computing the necessary conditions for the pair $(u_t, I)$ to minimize $E(u_t, I)$ and denoting $\pi := \frac{1}{\sigma^2}(g_1 I - I_1)^\flat$ yields
\begin{align*}
&\left\langle  DE(u_t, I), (\delta u_t, \delta I) \right\rangle = \int_0^1 \left\langle \frac{\delta \ell}{\delta u}(t), \delta u(t) \right\rangle dt
+ \frac{1}{\sigma^2} \langle I^\flat - I_0^\flat, \delta I \rangle
+ \langle \pi, \delta g_1 I + g_1 \delta I \rangle \\
&\qquad = \int_0^1 \left\langle \frac{\delta \ell}{\delta u}(t), \delta u(t) \right\rangle dt
+ \frac{1}{\sigma^2} \langle I^\flat - I_0^\flat + \sigma^2 g_1^{-1} \pi, \delta I \rangle
+ \left\langle \pi, g_1 \left( \int_0^1 \operatorname{Ad}_{g_t^{-1}} \delta u(t) dt \right) I \right\rangle \\
&\qquad = \int_0^1 \left\langle \frac{\delta \ell}{\delta u}(t) + g_t I \diamond g_{t,1} \pi, \delta u(t) \right\rangle dt
+ \frac{1}{\sigma^2} \langle I^\flat - I_0^\flat + \sigma^2 g_1^{-1} \pi, \delta I \rangle
\,.
\end{align*}
This leads to
\begin{align*}
\frac{\delta \ell}{\delta u}(t) &= -\, g_t I \diamond g_{t,1} \pi 
\,,\\
I^\flat &= I_0^\flat - \sigma^2 g_1^{-1} \pi 
\,,\\
\pi &= \frac{1}{\sigma^2}(g_1 I - I_1)^\flat.
\end{align*}
For images $I \in \mathcal{F}(\Omega, \mathbb{R})$ as in Section \ref{subsec-imagematching}, the equation for $I^\flat$ can be solved explicitly to find
\[
I = \frac{I_0 + \lvert D \phi_1 \rvert I_1 \circ \phi_1}{1+\lvert D\phi_1 \rvert}
\,.
\]
In this case $I$ constitutes a weighted average of $I_0$ and the deformed image $\phi_1^{-1} \cdot I_1$ at time $t=0$.
\end{example}

These examples all have a similar momentum map structure. The examples differed in the time point at which the inexactness of the matching was measured, or, as in the last case, in which of the images was being compared. We have restricted our attention primarily to only one of these possible formulations of LDM. However, the geometric interpretations are clearly similar in all cases and the momentum map plays the determining role in each case.

\section{Nonlinear Generalizations}
\label{nonlin_gen}

We now show that the formalism developed in \S\ref{subsec-AbFrm} generalizes easily to the case when the set of images is not necessarily a vector space and the cost function is not necessarily the Euclidean distance. This situation arises, for example, in the Landmark Matching Problem associated to points on the sphere for the study of neocortex, see \citet{MiTrYo2002} and references therein.

Suppose the set of images is a manifold $Q$ on which a group of transformation $G$ acts on the left. As before, we denote by $gI$ the action $g\in G$ on $I\in Q$. We consider a cost function of the form
\begin{equation}\label{action_manifolds}
E(u_t)=\int_0^1\ell(u_t)dt+ F\left(g^u_1 I_0,I_1\right),
\end{equation}
where $F$ is defined on $Q\times Q$. When $Q$ is a vector space $V$ with inner product norm $\|\cdot\|_V$, we recover the cost function \eqref{eq-abfrm-cost} by choosing
\[
F(I,J):=\frac{1}{2\sigma^2}\|I-J\|^2.
\]

The next theorem establishes the stationarity condition associated to the cost in \eqref{action_manifolds}.

\begin{theorem} Given a curve $t\mapsto u_t$ in the Lie algebra $\mathfrak{g}$ of $G$, we have
\[
DE(u_t)=0\;\Longleftrightarrow\;\frac{\delta \ell}{\delta u}(t)
=-\,\mathbf{J}\left(g^u_{t,1}\,\partial_1F(J^0_1,I_1)\right),
\]
where $\mathbf{J}:T^*Q\rightarrow\mathfrak{g}^*$ is the cotangent bundle momentum map and $\partial_1F(J^0_1,I_1)\in T^*_{J^0_1}Q$ is the tangent map to $F$ relative to the first variable. The momentum $\frac{\delta\ell}{\delta u}(t)$ satisfies the Euler-Poincar\'e equation
\[
\frac{d}{dt}\frac{\delta\ell}{\delta u}(t)=-\operatorname{ad}^*_{u_t}\frac{\delta\ell}{\delta u}(t).
\]
\end{theorem}
\begin{proof} The proof is similar to that of Theorem \ref{delta_E_left}. We will use the formula $\left\langle\mathbf{J}(\alpha_q),u\right\rangle=\left\langle\alpha_q,u_Q(q)\right\rangle$ for the momentum map $\mathbf{J}:T^*Q\rightarrow\mathfrak{g}^*$ associated to the cotangent lift action. Using Lemma \ref{delta_u_right}, we calculate
\begin{align*}
\left\langle DE(u_t),\delta u_t\right\rangle&=\delta \left(\int_0^1\ell(u(t))dt+F\left(g^u_1 I_0,I_1\right)\right)\\
&=\int_0^1\left\langle \frac{\delta \ell}{\delta u}(t),\delta u(t)\right\rangle dt+\left\langle \partial_1F(J^0_1,I_1),(\delta g^u_1) I_0\right\rangle\\
&=\int_0^1\left( \left\langle \frac{\delta \ell}{\delta u}(t),\delta u(t) \right\rangle dt +\left\langle \left(g^u_1\right)^{-1}\partial_1F(J^0_1,I_1),\left(\operatorname{Ad}_{g_{0,t}^u}\delta u(t)\right)_QI_0\right\rangle \right)dt \\
&=\int_0^1 \left(\left\langle \frac{\delta \ell}{\delta u}(t),\delta u(t)\right\rangle +\left\langle\mathbf{J}\left(\left(g^u_1\right)^{-1}\partial_1F(J^0_1,I_1)\right),\operatorname{Ad}_{g_{0,t}^u}\delta u(t)\right\rangle \right) dt\\
&=\int_0^1\left(\left\langle \frac{\delta \ell}{\delta u}(t)+\operatorname{Ad}^*_{g_{0,t}^u} \left(\mathbf{J}\left(\left(g^u_1\right)^{-1}\partial_1F(J^0_1,I_1)\right)\right),\delta u(t)\right\rangle\right)dt
\,,
\end{align*}
which must hold for all variations $\delta u(t)$. Therefore,
\begin{align*}
\frac{\delta \ell}{\delta u}(t)&=-\operatorname{Ad}^*_{g_{0,t}^u} \left(\mathbf{J}\left(\left(g^u_1\right)^{-1}\partial_1F(J^0_1,I_1)\right)\right)\\
&=-\,\mathbf{J}\left(g^u_{t,1}\,\partial_1F(J^0_1,I_1)\right)
\,,
\end{align*}
as required. The same proof as for Lemma \ref{EP} shows that the Euler-Poincar\'e equations are verified.
\end{proof}

When $Q$ is a vector space $V$, this stationarity condition can be rewritten equivalently by using the diamond map, $(\,\diamond\,)$, as
\[
\frac{\delta \ell}{\delta u}(t)=-J^0_t\diamond \left(g^u_{t,1}\,\partial_1F(J^0_1,I_1)\right).
\]

\paragraph{Landmark matching on manifolds.}
In the case of the Landmark Matching Problem on a Riemannian manifold $Q$, one chooses the cost function
\[
F(q_1,..,q_n;p_1,...,p_n):=\sum_{i=1}^n\frac{1}{2\sigma^2}d(q_i,p_i)^2,
\]
where $d$ is the Riemannian distance. This approach is used for imaging of the neocortex, where $Q$ is taken to be the sphere $S^2$. The energy to minimize has the form
\[
E(u_t)=\frac{1}{2}\int_0^1|u_t|^2_\mathcal{H}+\sum_{i=1}^n\frac{1}{2\sigma^2}d(\phi_1(q_i),p_i)^2,
\]
where $q_i, p_i\subset S^2$ are given.

\paragraph{LDM multimodal image matching.} The framework developed above allows us to understand geometrically the model developed in \citet{Vialard2009}, \S3.2. This model deals also with a change of intensity in the image $I:\Omega\rightarrow X$. This change of intensity can be modeled by an action $\eta\circ I$ of a diffeomorphism of the template co-domain $X$. In this case, the energy can have the general form
\[
E(v_t,w_t)=\int_0^1\ell(v_t,w_t)dt+F(\eta_1\circ I_0\circ\phi_1^{-1},I_1),
\]
where $\eta_t\in\operatorname{Diff}(X)$ and $\phi_t\in\operatorname{Diff}(\Omega)$ are the flows of $v_t$ and $w_t$, respectively. This problem can be recast in our formulation by considering the action of the direct product $\operatorname{Diff}(\Omega)\times\operatorname{Diff}(X)$ on the manifold $Q=\mathcal{F}(\Omega,X)$ given by
\[
(\phi,\eta)\cdot I:=\eta\circ I\circ\phi^{-1}.
\]
For simplicity, we suppose that $X$ is a vector space, but in general $X$ can be an arbitrary manifold. The cotangent lifted action on $\pi$ reads
\[
(\phi,\eta)\cdot (I,\pi) (x)= |\det D\phi^{-1}(x)|\,D\eta^{-1}(I(\phi(x)))^T\cdot \pi(\phi^{-1}(x))
\]
and the momentum map is
\[
\mathbf{J}(I,\pi)=\left(-\pi\boldsymbol{\cdot}\nabla I,\int_\Omega\pi(x)\delta_{I(x)}dx\right).
\]
Using these formulas, the stationarity condition is
\[
\frac{\delta\ell}{\delta v}=\pi_t\boldsymbol{\cdot}\nabla J^0_t,\qquad\frac{\delta \ell}{\delta w}=-\int_\Omega\pi_t(x)\delta_{J^0_t(x)}dx,
\]
where $J^0_t=\eta_t\circ I_0\circ \phi_t^{-1}$ and 
\[
\pi_t(x):=|\det D\phi^{-1}_{t,1}(x)|\left(D\eta_{t,1}(J^0_1(x))\right)^{-T}\partial_1F(J^0_1,I_1)(\phi_{t,1}^{-1}(x)).
\]
The last expression is obtained using the formula of the cotangent lifted action and the equality
\[
\left(D\eta_{t,1}^{-1}(J^0_t(\phi_{t,1}(x)))\right)^T=
D\eta_{t,1}(J^0_1(x))^{-T}.
\]
For more discussion, see \citet{Vialard2009}. 

\paragraph{Alternative approach.} We now consider an alternative approach that affects the geometric shape of the image $I:\Omega\rightarrow X$, as considered in \citet{Trouve1995}. This approach is different from that considered above. For example we can consider the case $X=S^2$ of images of unitary vectors in $\mathbb{R}^3$. In this case the shape can be modified by letting various groups of matrices act on $S^2$. These matrices are of course allowed to depend on the domain $\Omega$. We thus need to consider the group $\mathcal{F}(\Omega,G)$, where $G$ is a group acting on $X$. In order to also take into account the transformation on the domain, the semidirect product $\operatorname{Diff}(\Omega)\,\circledS\,\mathcal{F}(\Omega,G)\ni (\phi,
\theta)$ needs to be considered as in \citet{Trouve1995}. This group acts in a natural way on the space $\mathcal{F}(\Omega,X)$ of images via the left action
\[
(\phi,\theta)\cdot I=(\theta I)\circ\phi^{-1},
\]
where the function $\theta I$ is defined by $(\theta I)(x):=\theta(x)I(x)$ and in the last term we use the $G$-action on $X$.
A vector field on this Lie algebra has components $(u,\nu)$ where $u$ is a vector field on $\Omega$ and $\nu:\Omega\rightarrow\mathfrak{g}$. Using the multiplication rule $(\phi,\theta)(\bar\phi,\bar\theta)=(\phi\circ\bar\phi,(\theta\circ\bar\phi)\bar\theta)$ in the semidirect product, the ODE $\partial_t(\phi_t,\theta_t)=(u_t,\nu_t)(\phi_t,\theta_t)$ reads
\[
\dot\phi_t=u_t\circ\phi_t,\qquad \dot\theta_t=(\nu_t\circ\phi_t)\theta_t,\qquad \phi_0=e, \qquad\theta_0=e.
\]
For simplicity, we suppose that $X$ is a vector space. The infinitesimal action on the space of images reads $(u,\nu)I=\nu I-\nabla I\boldsymbol{\cdot}u$ hence the cotangent bundle momentum map is
\[
\mathbf{J}(I,\pi)=\left(-\pi\boldsymbol{\cdot}\nabla I,I\diamond \pi\right),
\]
where $I\diamond \pi$ is the function with values in $\mathfrak{g}^*$ defined by $(I\diamond \pi)(x)=I(x)\diamond \pi(x)$ and the diamond on the right denotes the momentum map associated to the action of $G$ on $X$. In order to formulate the stationarity condition, we also need the expression of the cotangent lifted action given by
\[
(\phi,\theta)\cdot (I,\pi)=\left((\theta I)\circ\phi^{-1},|\det D\phi^{-1}|\,(\theta\pi)\circ\phi^{-1}\right).
\]
The cost function has the form
\[
E(u_t,\nu_t)=\int_0^1\ell(u_t,\nu_t)dt-F\left((\theta_1I_0)\circ\phi_1^{-1},I_1\right).
\]
The stationarity conditions are thus given by
\[
\frac{\delta\ell}{\delta u}(t)=\pi_t\boldsymbol{\cdot}\nabla J^0_t,\qquad \frac{\delta \ell}{\delta\nu}(t)=-J^0_t\diamond \pi_t,
\]
where $J^0_t=(\phi_t,\theta_t)\cdot I_0=(\theta_t I_0)\circ\phi_t^{-1}$ and
\[
\pi_t=|\det D\phi_{1,t}^{-1}|\left(\theta_{1,t}\partial_1F(J^0_1,I_1)\right)\circ\phi_{1,t}^{-1}.
\]
For example, when $F(I,J)=\frac{1}{2\sigma^2}\|I-J\|_{L^2}$, relative to an inner product on $X$, then $\partial_1F(I,J)=\frac{1}{\sigma}(I-J)^\flat\in\mathcal{F}(\Omega,X^*)$, where $\flat$ is associated to the inner product on $X$. In this case, the stationarity conditions are
\[
\frac{\delta \ell}{\delta u}(t)=|\det D\phi_{1,t}^{-1}|\,(J^0_t-J^1_t)^\flat\nabla J^0_t,\qquad \frac{\delta\ell}{\delta\nu}=J^0_t\diamond |\det D\phi_{1,t}^{-1}|\,(J^0_t-J^1_t)^\flat.
\]

\section{Conclusions}
\label{sec-conclusions-outlook}

This paper has revealed that Beg's algorithm from \citet{Beg2003} and \citet{Begetal2005} for image registration in the LDM framework is  the cotangent-lift momentum map associated to the action of diffeomorphisms on scalar functions. Accordingly, the momentum map has emerged as a central organizing principle in the abstract framework inspired by image registration. The momentum map provides the means of unifying the LDM approach for the registration of different data structures that use different penalty terms and different Lie groups. Different data structures summon different group actions to define their transformations and they will therefore give rise to different momentum maps. But once the momentum map is computed, it is straight-forward to implement the corresponding gradient-descent scheme for image registration. 
{
The momentum map systematically incorporates both the specification of distance on the space of images and the transformation properties of their data structure. 
}

Exploring the specification of distance and dealing with other data structures has been left for future work. For example, the pioneering work of \citet{Alex-etal-2001} and \citet{Cao2006} on the registration of DT-MRIs led to the action on symmetric tensors discussed in Section \ref{DT-MRI}. We plan to compare the momentum map for this action with the usual push-forward action on tensor fields to gain further insights into the matching procedures for tensor data structures.

{
The advantage of our method in practical applications is that it systematizes the development of algorithms for registering images in various types of data structure, by identifying the momentum map as the shared fundamental element for registration of images in any data structure. This means, for example, that registration of multi-channel or multi-modal images can be accomplished simply by applying the present method to the \emph{sum} of  momentum maps for the different types of data structures.
}

Images encountered in applications often contain information at several length scales. A heuristic approach for adapting the registration procedure to take into account these length scales suggested replacing the kernel in \eqref{eq-mot-grad} by the sum of two kernels $K_{\alpha_1} + K_{\alpha_2}$, with two different length scales $\alpha_1$ and $\alpha_2$ for their corresponding filters. We have shown that this strategy has a geometric interpretation. Namely, instead of using a single diffeomorphism group to perform image registration, we can use the semidirect product of two such groups, each associated to its own length scale, {
the larger one \emph{sweeping} the smaller one by semidirect-product action.
}
The resulting equations \eqref{kernel-sum} then \emph{coincide} with the sum-of-kernels strategy. Similarly, the same result could be obtained for the sum of three and more kernels. 
{
Recently, this sum-of-kernels strategy for registration has been applied successfully in \cite{RiViWoHoRu2010} for measurement of atrophy of tissues in the hippocampus due to Alzheimer's disease. This result opens new perspectives in clinical applications of multi-resolution imaging.
}

Other formulations of LDM that were intended to make the registration symmetric, as proposed by \citet{Avants2008}, \citet{Beg2007} and \citet{ZaHaNi2009}, were also discussed and written geometrically. We have shown that all these cases exhibit similar momentum map structures. The main differences arise from the choice of the time at which the momentum map is to be evaluated. Once again, the momentum map appears as a unifying framework allowing systematic comparisons among the different examples.

We have also explored a natural generalization of the framework to incorporate data structures living in manifolds, which do not have the linear structures of vector fields. Examples included landmarks on a sphere. Since in this case no norm is available to measure distances between two images, a distance function must be chosen. Further applications and capabilities of this nonlinear framework will be explored in future work.

\bibliographystyle{plainnat}

\end{document}